\documentclass[11pt]{amsart}
\usepackage{slashed}
\usepackage{hyperref}
\usepackage{amsthm}
\usepackage{epsfig,subfigure,latexsym,dcolumn,amsmath,longtable,setspace}
\usepackage{amsmath}
\usepackage{amsfonts}
\usepackage{amssymb}
\usepackage{latexsym}
\usepackage{mathrsfs}
\usepackage{amsmath,amscd}
\usepackage{amsmath, amscd, amssymb}
\usepackage{graphicx, psfrag}
\usepackage[frame,cmtip,arrow,matrix,line,graph,curve,color]{xy}
\usepackage{graphpap, color}
\usepackage[mathscr]{eucal}
\usepackage{mathrsfs}
\usepackage{tikz}
\usepackage{xcomment}
\usepackage{verbatim}

\usetikzlibrary{matrix,arrows}

\numberwithin{equation}{section}

\newtheorem{prop}{Proposition}[section]
\newtheorem{theo}[prop]{Theorem}
\newtheorem{lemm}[prop]{Lemma}

\newtheorem{defi}[prop]{Definition}

\newcommand{\vv}{\vert\vert}

\title[Growing solutions for the Klein-Gordon equation on Kerr]{Exponentially growing finite energy solutions for the Klein-Gordon equation on sub-extremal Kerr spacetimes}
\author{Yakov Shlapentokh-Rothman}
\thanks{This work was partially supported by NSF grant DMS-0943787.}
\address{Department of Mathematics, MIT, Cambridge, MA 02139, USA}
\email{yakovsr@math.mit.edu}
\begin{document}
\begin{abstract}
For any sub-extremal Kerr spacetime with non-zero angular momentum, we find an open family of non-zero masses for which there exist smooth, finite energy, and exponentially growing solutions to the corresponding Klein-Gordon equation. If desired, for any non-zero integer $m$, an exponentially growing solution can be found with mass arbitrarily close to $\frac{\left|am\right|}{2Mr_+}$. In addition to its direct relevance for the stability of Kerr as a solution to the Einstein-Klein-Gordon system, our result provides the first rigorous construction of a superradiant instability. Finally, we note that this linear instability for the Klein-Gordon equation contrasts strongly with recent work establishing linear stability for the wave equation.
\end{abstract}
\maketitle

\tableofcontents
\section{Introduction}
The Kerr spacetime $(\mathcal{M},g_{a,M})$ is a two parameter family of asymptotically flat, stationary, and axisymmetric solutions to the vacuum Einstein equations Ric$(g) = 0$. As a precursor to establishing the conjectured non-linear stability of Kerr, there has been much study of the linear stability problem for various equations on a fixed Kerr background. In this paper we will study the Klein-Gordon equation:
\[\left(\Box_g - \mu^2\right)\psi = 0.\]
Here $\mu \geq 0$ is the mass of the scalar field $\psi$. In contrast to previous works on the wave equation ($\mu = 0$) showing linear stability, we will produce an open family of masses for which the Klein-Gordon equation exhibits linear \emph{instability}, i.e.~for these masses there exists smooth, finite energy solutions which grow exponentially in time.

Recall that in the domain of outer communication we can parameterize the Kerr spacetime with Boyer-Lindquist coordinates $(t,r,\theta,\phi) \in \mathbb{R} \times (M + \sqrt{M^2-a^2},\infty) \times \mathbb{S}^2$ where the metric takes the form
\[g_{a,M} = -\left(1-\frac{2Mr}{\rho^2}\right)dt^2 - \frac{4Mar\sin^2\theta}{\rho^2}dtd\phi + \frac{\rho^2}{\Delta}dr^2 + \rho^2 d\theta^2 + \sin^2\theta\frac{\Pi}{\rho^2}d\phi^2,\]
\[r_{\pm} := M \pm \sqrt{M^2-a^2},\]
\[\Delta := r^2 - 2Mr + a^2 = (r-r_+)(r-r_-),\]
\[\rho^2 := r^2 + a^2\cos^2\theta,\]
\[\Pi := (r^2+a^2)^2 - a^2\sin^2\theta\Delta.\]
Our main result is
\begin{theo}\label{mainResult}Fix a Kerr spacetime $(\mathcal{M},g_{a,M})$ with $M > 0$ and $0 < |a| < M$. Then there exists an open family of masses $\mu$ with $\epsilon_{\mu} > 0$ and a non-zero, smooth, and finite energy solution $\psi$ to the corresponding Klein-Gordon equation
\[\left(\Box_g-\mu^2\right)\psi = 0\]
such that for every $(t,r,\theta,\phi) \in \mathbb{R}\times(M+\sqrt{M^2-a^2},\infty)\times\mathbb{S}^2$
\begin{equation}\label{growth}
e^{\epsilon_{\mu}t}\left|\partial^{\alpha}\psi(0,r,\theta,\phi)\right| \lesssim_{\alpha} \left|\partial^{\alpha}\psi(t,r,\theta,\phi)\right| \text{ for all multi-indices }\alpha.
\end{equation}
These statements should be understood with respect to Boyer-Lindquist coordinates. For every non-zero integer $m$, $\mu$ can be chosen arbitrarily close to $\frac{|am|}{2Mr_+}$. In particular, $\mu$ can be made arbitrarily small as $a\to 0$.
\end{theo}

Remark: For convenience, we have stated our theorem in Boyer-Lindquist coordinates; however, these coordinates break down on the future event horizon $\mathcal{H}^+$ (see section \ref{coord}). Nevertheless, it will be easy to see that along $\mathcal{H}^+$ the solutions constructed are also exponentially growing with respect to the regular $t^*$ coordinate; see the discussion in section \ref{boundaryConditions}.

Theorem \ref{mainResult} may suggest that the Kerr spacetime is non-linearly \emph{unstable} as a solution to the Einstein-Klein-Gordon system.\footnote{The Einstein-Klein-Gordon system for a spacetime $(\mathcal{M},g)$ and massive scalar field $\psi$ is \[\text{Ric}_{\alpha\beta}(g) - \frac{1}{2}\text{R}(g)g_{\alpha\beta} = 8\pi\mathbf{T}_{\alpha\beta}\left(g,\psi\right),\]
\[\left(\Box_g-\mu^2\right)\psi = 0.\]
Here $R(g)$ is the scalar curvature, and $\mathbf{T}_{\alpha\beta}$ is the energy-momentum tensor (\ref{energyMom}).} Additionally, our result provides the first rigorous construction of a superradiant instability. Informally put, superradiance can occur in a black hole spacetime when there does not exist a globally defined Killing vector field which is both timelike or null at infinity and timelike or null on the horizon. For such spacetimes ``energy'' may radiate out of the black hole and, depending on the particular dynamics under consideration, lead to a \emph{superradiant} instability.

To make these ideas more concrete, let's focus on the Klein-Gordon equation and begin by briefly recalling the energy-momentum tensor formalism (see \cite{n3} for a proper introduction). Let $g$ denote the (Lorentzian) metric on our spacetime and $\nabla$ denote covariant differentiation. For any function $\psi$ we define the energy-momentum tensor
\begin{equation}\label{energyMom}
\mathbf{T}_{\alpha\beta} := \text{Re}\left(\nabla_{\alpha}\psi\overline{\nabla_{\beta}\psi}\right) - \frac{1}{2}g_{\alpha\beta}\left(\left|\nabla\psi\right|^2 + \mu^2\left|\psi\right|^2\right).
\end{equation}
For any vector field $X$ we define a corresponding $1$-form, called a ``current,'' by
\begin{equation}\label{current}
\mathbf{J}^X_{\alpha} := \mathbf{T}_{\alpha\beta}X^{\beta}.
\end{equation}
The key identity is
\begin{equation}\label{divergence}
\nabla^{\alpha}\mathbf{J}^X_{\alpha} = \text{Re}\left(\left(\nabla^{\alpha}\nabla_{\alpha}\psi-\mu^2\psi\right)\overline{\left(X\psi\right)}\right) + \frac{1}{2}\mathbf{T}_{\alpha\beta}\pi^{\alpha\beta}.
\end{equation}
Here $\pi$ denotes the deformation tensor of $X$:
\[\pi^{\alpha\beta} := \nabla^{\alpha}X^{\beta} + \nabla^{\beta}X^{\alpha}.\]
This vanishes if and only if $X$ is Killing. In particular, if $\psi$ solves the Klein-Gordon equation and $X$ is Killing, we find that $\mathbf{J}^X_{\alpha}$ is divergence free. In this case, for any two homologous hypersurfaces $\Sigma_1$ and $\Sigma_2$, the divergence theorem gives a conservation law:
\begin{equation}\label{conserve}
\int_{\Sigma_1}\mathbf{J}^X_{\alpha}n^{\alpha}_{\Sigma_1} = \int_{\Sigma_2}\mathbf{J}^X_{\alpha}n^{\alpha}_{\Sigma_2}.
\end{equation}
Here $n_{\Sigma_i}$ denotes the (future oriented) normal to the hypersurface $\Sigma_i$, and the integrals are with respect to the natural volume forms (the ones that make the divergence theorem true). For the identity (\ref{conserve}) to be useful, we need some positivity of $\mathbf{J}^X_{\alpha}n^{\alpha}_{\Sigma_i}$. One may show (see \cite{n3}) that $\mathbf{J}^X_{\alpha}n^{\alpha}_{\Sigma_i}$ is a positive definite non-degenerate quadratic form in $\psi$ and its derivatives at a point $x_0$ if and only if $X$ and $n_{\Sigma_i}$ are timelike and future directed at $x_0$. If we allow the quadratic form to be degenerate, then we may allow $X$ and $n_{\Sigma_i}$ to also be null. The significance of superradiance for stability problems should now be clear.

Of course, the canonical example of a spacetime admitting superradiance is Kerr.\footnote{Note that any spacetime of the form $(\mathbb{R} \times M,-dt^2 + g_M)$, with $(M,g_M)$ a Riemannian manifold, admits a globally defined timelike Killing vector field corresponding to time translation. Thus, superradiance is a truly Lorentzian phenomenon.

After the Kerr spacetime, the next most interesting examples from the instability point of view are perhaps ``small'' Kerr-AdS black holes \cite{n20}, \cite{n21}.} The geometry of Kerr is reviewed in section \ref{kerr}; for now, let us simply recall that there is a unique (up to normalization) Killing vector field which is future directed and timelike at infinity and that this vector field is spacelike on (almost all of) the horizon. If we let $T$ denote this vector field and $\Phi$ denote the unique (up to normalization) future directed Killing vector field which vanishes along the axis of symmetry, then the null generator of the horizon $\mathcal{H}^+$ is
\[L := T + \omega_+\Phi\]
where $\omega_+ := \frac{a}{2Mr_+}$ is the ``angular velocity'' of the black hole. One then finds that the energy density along the horizon for a solution $\psi$ to the Klein-Gordon equation is
\begin{equation}\label{superHor}
\mathbf{J}^T_{\alpha}L^{\alpha} = \text{Re}\left(T\psi\overline{L\psi}\right) = \text{Re}\left(T\psi\overline{\left(T\psi + \frac{a}{2Mr_+}\Phi\psi\right)}\right).
\end{equation}
When the black hole possesses non-zero angular momentum ($a \neq 0$) it is clearly possible for this quantity to be negative, and thus, in principle, energy can radiate out of the black hole.

Our exponentially growing solutions will be superradiant bound states,\footnote{Recall that it only makes sense to consider bound state solutions when the mass of the scalar field is non-zero.} i.e.~the energy flux will be negative along the horizon, and the solution will spatially decay exponentially fast so that no energy is radiated away to infinity. For such solutions the energy coming out of the black hole cannot escape; this is what provides the mechanism for the exponential growth. For technical reasons we will first construct bound state solutions which have exactly zero energy flux on the horizon. As one expects, these solutions will neither grow nor decay. Then, a perturbation argument will produce the exponentially growing solutions.

\subsection{Linear Stability of The Wave Equation}
The wave equation is simply the Klein-Gordon equation with $\mu = 0$. It is instructive to contrast our instability result for the Klein-Gordon equation with previous work showing that the wave equation is linearly stable. Due to the difficulties of superradiance and the complicated trapping structure of Kerr, the most basic boundedness and decay statements for the wave equation on the Kerr spacetime remained unresolved until the quite recent \cite{n2} (we should also mention the previous \cite{n45} and \cite{n46}). This breakthrough followed a lot of earlier work restricted to Kerr spacetimes with $|a| \ll M$, e.g. \cite{n15}, \cite{n5}, \cite{n18}, \cite{n40}, \cite{n41}, \cite{n42}, \cite{n43}, and \cite{n44}.

The instability result of this paper serves to the emphasize the subtle effects superradiance can have on the linear stability problem and helps to ``explain'' why even establishing boundedness for the wave equation on Kerr is difficult. In particular, since the Klein-Gordon equation decays faster than the wave equation on Minkowski space, one may have expected that the Klein-Gordan equation would be easier to control. However, as $|a| \to 0$ (where superradiance is \emph{weaker} and one expects the problem to get \emph{easier}) we have produced exponentially growing and finite energy solutions with arbitrarily small mass. Thus, any argument used for the wave equation must break down for Klein-Gordon equations with arbitrarily small mass. On a more conceptual level, we see that as one passes into the relativistic world, new obstructions to \emph{boundedness}, not just decay, arise in the superradiant bounded-frequency regime.

Of course, the sub-extremal Kerr spacetime is far from the only background on which to study the wave or Klein-Gordon equations; in fact, it is quite interesting to explore how changing the black hole geometry affects the subtle interplay between trapping, superradiance, and the redshift. In the sequence of works \cite{n56}, \cite{n57}, and \cite{n58}, Holzegel and Holzegel-Smulevici established a logarithmic upper and \emph{lower} bound on the decay rate for the wave and Klein-Gordon equations on non-superradiant Schwarzschild/Kerr-AdS spacetimes.\footnote{Here non-superradiant means that there exists a global timelike Killing vector field. In particular, it is possible to immediately rule out solutions of the type constructed in this paper.} The slow decay rate is directly traceable to a stable trapping phenomenon. In a series of papers \cite{n59}, \cite{n60}, \cite{n61}, \cite{n62}, and \cite{n63} Aretakis has studied the wave equation on various \emph{extreme} black holes,\footnote{The extreme Kerr spacetime occurs when $|a| = M$.} where there is a loss of the redshift due to the vanishing surface gravity of the horizon. One of the most striking results obtained is that even within the context of \emph{axisymmetric} solutions to the wave equation on extremal Kerr, for which there is no superradiance, second derivatives of the solution blow up along the horizon.\footnote{Interestingly, it was also shown that the solutions itself decays in time.} Taken together with the results of this paper, these various ``instabilities'' serve to emphasize the miraculous properties of the wave equation on sub-extremal Kerr.

\subsection{The Geometry of Kerr: The Ergoregion and Superradiance}\label{kerr}
In this section we shall briefly review the relevant aspects of the geometry of Kerr. For a true introduction to the Kerr spacetime we recommend \cite{n4} and \cite{n13}. The lecture notes \cite{n5} provide a good introduction to the interaction between the geometry of Kerr and the behavior of linear waves.
\subsubsection{Coordinate Systems}\label{coord}
Outside the black hole, the Kerr metric in Boyer-Lindquist coordinates $(t,r,\theta,\phi) \in \mathbb{R} \times (r_+,\infty) \times \mathbb{S}^2$ is given by
\[g = -\left(1-\frac{2Mr}{\rho^2}\right)dt^2 - \frac{4Mar\sin^2\theta}{\rho^2}dtd\phi + \frac{\rho^2}{\Delta}dr^2 + \rho^2 d\theta^2 + \sin^2\theta\frac{\Pi}{\rho^2}d\phi^2,\]
\[r_{\pm} := M \pm \sqrt{M^2-a^2},\]
\[\Delta := r^2 - 2Mr + a^2 = (r-r_+)(r-r_-),\]
\[\rho^2 := r^2 + a^2\cos^2\theta,\]
\[\Pi := (r^2+a^2)^2 - a^2\sin^2\theta\Delta.\]
There are two free parameters $M$ and $a$. The first parameter $M$ is the mass of the black hole, and $aM$ is the angular momentum. The $1$-parameter family where $a = 0$ is known as the ``Schwarzschild'' spacetime. Some non-zero mass $M$ of the spacetime will be fixed throughout the paper, so the term ``mass'' will always refer to $\mu$. For various reasons relating to the global geometry of Kerr, physically relevant Kerr spacetimes must satisfy $\left|a\right| \leq M$ \cite{n4}. In this paper, we shall study sub-extremal Kerr, which corresponds to $0 < \left|a\right| < M$.\footnote{One expects that the instability result of this paper also holds on extremal Kerr where $|a| = M$; however, the relevant equations have a different structure at the horizon which precludes drawing any immediate conclusions from the sub-extremal case.} This assumption guarantees that $r_{\pm}$ both exist and are distinct.

Though Boyer-Lindquist coordinates are often convenient, they break down when $r = r_+$. We shall thus introduce another coordinate system. Let us define two functions $\overline{t}(r)$ and $\overline{\phi}(r)$ on $(r_+,\infty)$ up to a constant by
\[\frac{d\overline{t}}{dr} := \frac{r^2+a^2}{\Delta}\]
\[\frac{d\overline{\phi}}{dr} := \frac{a}{\Delta}.\]
Then we define Kerr-star coordinates $(t^*,r,\theta,\phi^*)$ by
\[t^*(t,r) := t + \overline{t}(r)\]
\[\phi^*(\phi,r) := \phi + \overline{\phi}(r).\]
In these coordinates the metric becomes
\[g = -\left(1-\frac{2Mr}{\rho^2}\right)(dt^*)^2 - \frac{4Mar\sin^2\theta}{\rho^2}dt^*d\phi^* + 2dt^*dr + \]
\[\rho^2 d\theta^2 + \sin^2\theta\frac{\Pi}{\rho^2}(d\phi^*)^2 -2a\sin^2\theta drd\phi^*.\]
Note that we can now allow $(t^*,r,\theta,\phi^*) \in \mathbb{R}\times (0,\infty)\times \mathbb{S}^2$. The future event horizon is defined to be the null hypersurface $\{r = r_+\}$. This is the boundary of the black hole. We call the region $\{r > r_+\}$ the ``domain of outer communication.'' Lastly, we note that in their common domain, $\partial_t$ in Boyer-Lindquist coordinates is equal to $\partial_{t^*}$ in Kerr-star coordinates. A similar statement applies to $\partial_{\phi}$ and $\partial_{\phi^*}$.
\subsubsection{The Ergoregion and Superradiance}\label{ergoSuper}
On Minkowski space, the Killing vector field $T := \partial_t$ is everywhere timelike. Combining this with the energy-momentum formalism \cite{n3} immediately implies that $\vv\psi(t)\vv_{\dot{H}_x^1}^2 + \left\vert\left\vert\dot{\psi}(t)\right\vert\right\vert_{L^2_x}^2+ \mu^2\vv\psi(t)\vv_{L_x^2}^2$ is constant in time. As in Minkowksi space, in Kerr we set $T := \partial_t$. One finds that $T$ is the unique (up to normaliation) Killing vector field which is future directed and timelike for all large $r$. Unfortunately, when
\[\Delta - a^2\sin^2\theta < 0\]
then $T$ is spacelike. This region is known as the ergoregion. Due to the ergoregion, the conserved quantity associated to $T$ has no definite sign. Even more disturbing than this loss of a ``free'' boundedness statement is the Penrose process \cite{n4},\cite{n24},\cite{n25}. This is a thought experiment where a test particle exploits the ergoregion to extract energy from the black hole. In the introduction we have already explained how to see this energy extraction at the level of the Klein-Gordon equation (\ref{superHor}).

It is important to keep in mind that superradiance cannot occur if either $a$ or $\Phi\psi$ vanishes. Lastly, we would like to emphasize that the problems of the ergoregion and superradiance occur for the wave equation. Hence, the presence of these two features alone certainly does not imply linear instability.
\subsection{Precise Statement of Results}
We will rigorously construct finite energy solutions to the Klein-Gordon equation
\[\left(\Box_g - \mu^2\right)\psi = 0\]
on sub-extremal Kerr which grow exponentially. These growing solutions will be ``mode solutions'' of the form
\begin{equation}\label{modeEqn}
\psi(t,r,\theta,\phi) := e^{-i\omega t}e^{i m\phi}S_{\kappa ml}(\theta)R(r)
\end{equation}
where $\omega \in \mathbb{C},\ m \in \mathbb{Z},\ l \in \mathbb{Z}_{\geq |m|},\ \text{ and }\kappa := a^2\left(\omega^2-\mu^2\right)$. Here $(t,r,\theta,\phi)$ are Boyer-Lindquist coordinates where, as is well known, the Klein-Gordon equation separates. The functions $S_{\kappa ml}$ and $R$ must satisfy appropriate ordinary differential equations and boundary conditions (see section \ref{modeStuff}) so that $\psi$ solves the Klein-Gordon equation, extends smoothly to the horizon (where Boyer-Lindquist coordinates break down), and has finite energy (and finite higher order energies) along asymptotically flat Cauchy hypersurfaces which intersect the future event horizon. For our mode solution to grow with time, we must have Im$(\omega) > 0$. Such solutions are called ``unstable modes.'' We say that these modes ``lie in the upper half-plane.'' Mode solutions with $\omega \in \mathbb{R}$ will be called ``real modes.'' We say that these modes ``lie on the real axis.'' It will be convenient to refer to the tuple $(\omega,m,l,\mu)$ as the ``parameters'' of the mode. Lastly, we observe that (\ref{superHor}) implies that a mode solution exhibits superradiance if and only if
\begin{equation}\label{sup}
am\text{Re}\left(\omega\right) - 2Mr_+\left|\omega\right|^2 > 0.
\end{equation}

We can now state our main result:
\begin{theo}\label{blackHoleBomb}Fix a sub-extremal Kerr spacetime with mass $M$ and angular momentum $aM$. Let $m \in \mathbb{Z}$ and $\omega_R(0) \in \mathbb{R}$ satisfy $am-2Mr_+\omega_R(0) = 0$ and $am \neq 0$. Then, for each $l \in \mathbb{Z}_{\geq |m|}$ and sufficiently small $\delta > 0$, there exists $\mu(0) > |\omega_R(0)|$, real analytic $\omega_R(\epsilon)$, and real analytic $\mu(\epsilon)$ such that for every $-\delta < \epsilon < \delta$, there exists a mode solution with parameters $(\omega_R(\epsilon) + i\epsilon,m,l,\mu(\epsilon))$. As $l \to \infty$, $\mu(0)$ will converge to $\omega_R(0)$. Lastly, these unstable modes must all be superradiant\footnote{One may easily check that (\ref{strongsuper}) is stronger than (\ref{sup}) via the inequality $\frac{x^2+y^2}{|x|} \geq \sqrt{x^2+y^2}$.}
\begin{equation}\label{strongsuper}
|am| - 2Mr_+\sqrt{\omega_R^2(\epsilon) + \epsilon^2} > 0
\end{equation}
and lose mass as they become unstable
\[\frac{\partial\mu}{\partial\epsilon}(0) < 0.\]
\end{theo}
Here is a picture of the values $\{\omega(\epsilon)\} \in \mathbb{C}$ traced out by the various $1$-parameter families of modes associated to a fixed $l$:

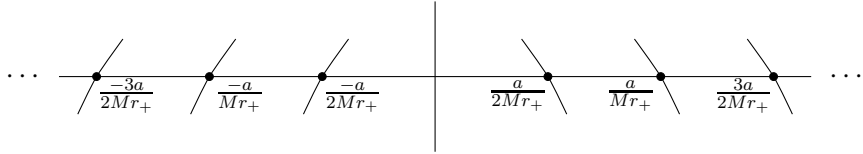
\begin{center}
\begin{tikzpicture}
    \draw (-5,0) -- (5,0);
    \draw (0,-1) -- (0,1);
    \node at (1.5,0) [circle, draw = black, fill = black, scale =  .1mm]{};
    \node at (1.1,-.25) [scale = .3mm]{$\frac{a}{2Mr_+}$};
    \draw plot [smooth] coordinates{(1.75,-.5) (1.5,0) (1.15,.5)};
    \node at (3,0) [circle, draw = black, fill = black, scale =  .1mm]{};
    \node at (2.6,-.25) [scale = .3mm]{$\frac{a}{Mr_+}$};
    \draw plot [smooth] coordinates{(3.25,-.5) (3,0) (2.65,.5)};
    \node at (4.5,0) [circle, draw = black, fill = black, scale =  .1mm]{};
    \node at (4.1,-.25) [scale = .3mm]{$\frac{3a}{2Mr_+}$};
    \draw plot [smooth] coordinates{(4.75,-.5) (4.5,0) (4.15,.5)};
    \node at (-1.5,0) [circle, draw = black, fill = black, scale =  .1mm]{};
    \node at (-1.1,-.25) [scale = .3mm]{$\frac{-a}{2Mr_+}$};
    \draw plot [smooth] coordinates{(-1.75,-.5) (-1.5,0) (-1.15,.5)};
    \node at (-3,0) [circle, draw = black, fill = black, scale =  .1mm]{};
    \node at (-2.6,-.25) [scale = .3mm]{$\frac{-a}{Mr_+}$};
    \draw plot [smooth] coordinates{(-3.25,-.5) (-3,0) (-2.65,.5)};
    \node at (-4.5,0) [circle, draw = black, fill = black, scale = .1mm]{};
    \node at (-4.1,-.25) [scale = .3mm]{$\frac{-3a}{2Mr_+}$};
    \draw plot [smooth] coordinates{(-4.75,-.5) (-4.5,0) (-4.15,.5)};
    \node at (5.5,0) []{$\cdots$};
    \node at (-5.45,0) []{$\cdots$};
\end{tikzpicture}
\end{center}
The reader should keep in mind that we have \emph{not} produced any estimates for the lengths of these curves.

For each choice of $m \in \mathbb{Z}\setminus\{0\}$, there is a countable family of \emph{intervals} of masses $\mu$ associated to growing solutions (indexed by $l$). These intervals will have an accumulation point at $\frac{|am|}{2Mr_+}$. The following picture may be useful for visualization:
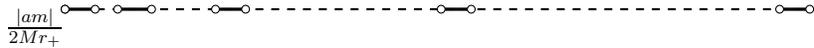
\begin{center}
\begin{tikzpicture}
\draw [dashed, thick] (0,0) -- (10,0);
\node at (0,0) [circle, draw = black, fill = white, scale = .1mm]{};
\node at (-.4,-.25) [scale = .3mm]{$\frac{\left|am\right|}{2Mr_+}$};

\node at (0,0) [circle, draw = black, fill = white, scale =  .1mm]{};
\draw [very thick] (.05,0) -- (.4,0);
\node at (.4,0) [circle, draw = black, fill = white, scale = .1mm]{};
\node at (.7,0) [circle, draw = black, fill = white, scale =  .1mm]{};
\draw [very thick] (.75,0) -- (1.15,0);
\node at (1.15,0) [circle, draw = black, fill = white, scale = .1mm]{};
\node at (2,0) [circle, draw = black, fill = white, scale =  .1mm]{};
\draw [very thick] (2.05,0) -- (2.4,0);
\node at (2.4,0) [circle, draw = black, fill = white, scale = .1mm]{};
\node at (5,0) [circle, draw = black, fill = white, scale =  .1mm]{};
\draw [very thick] (5.05,0) -- (5.4,0);
\node at (5.4,0) [circle, draw = black, fill = white, scale = .1mm]{};
\node at (9.5,0) [circle, draw = black, fill = white, scale =  .1mm]{};
\draw [very thick] (9.55,0) -- (9.9,0);
\node at (9.9,0) [circle, draw = black, fill = white, scale = .1mm]{};
\end{tikzpicture}
\end{center}
Lest the reader be misled, we emphasize that we do \emph{not} have any estimates for how large these intervals are, and (despite the picture) we have \emph{not} proven that we can find $\epsilon > 0$ such that the interval $\left(\frac{\left|am\right|}{2Mr_+},\frac{\left|am\right|}{2Mr_+} + \epsilon\right)$ is entirely made up of unstable masses. However, in light of the arguments in section \ref{intoHalfPlane} we would certainly conjecture that this last statement is true.

The construction of the exponentially growing modes is achieved by perturbing modes corresponding to real $\omega$. Thus, before proving Theorem \ref{blackHoleBomb}, we will undertake an analysis of modes corresponding to real $\omega$. For these modes we have two main results. The first is an existence result (already contained in Theorem \ref{blackHoleBomb}). The second shows that the assumptions on the frequency parameters from Theorem \ref{blackHoleBomb} are necessary.
\begin{theo}\label{restrictions}Suppose there exists a mode solution with parameters $(\omega,m,l,\mu)$ such that $\omega \in \mathbb{R}$ and $\mu^2 > \omega^2$. Then the following statements are true.
\begin{enumerate}
    \item\label{part1} We have $am - 2Mr_+\omega = 0$.
    \item\label{part2} We have $am \neq 0$.
    \item\label{part3} There exists a function $C(\omega,m,l)$ such that $\omega^2 < \mu^2 < \omega^2 + C(\omega,m,l)$ and \[\lim_{l\to\infty}C(\omega,m,l) = 0.\]
\end{enumerate}
\end{theo}
We will close the section with two remarks. First, we note that we can rephrase the condition $am-2Mr_+\omega = 0$ more geometrically. Let $L$ be a null generator of the horizon, e.g. $T + \frac{a}{2Mr_+}\Phi$.  Then
\[am-2Mr_+\omega = 0 \Leftrightarrow L\psi = 0 \Leftrightarrow \text{ No energy flux along the horizon}\]
Thus, our real mode solutions are simply solutions to the Klein-Gordon equation with exactly vanishing energy flux along the horizon.

Second, an appropriately modified version of Theorem \ref{restrictions} also holds if $\omega^2 > \mu^2$. Instead of requiring finite energy, one should enforce the outgoing condition
\[R \sim \frac{e^{i\sqrt{\omega^2-\mu^2}r^*}}{r}\text{ at }\infty.\]
Here $r^*$ is defined up to a constant by
\[\frac{dr^*}{dr} = \frac{r^2+a^2}{\Delta}.\]
Though we will not pursue this here, one can rule out such solutions by modifying the arguments of \cite{n23}.
\subsection{Previous Works on Mode Solutions}
To the best of the author's knowledge, there are no previous rigorous constructions of growing solutions for the Klein-Gordon equation. However, there are a few results which rule out growing modes in certain parameter ranges. In \cite{n14} Beyer showed that no unstable modes can exist if
\[\mu \geq \frac{\left|am\right|}{2Mr_+}\sqrt{1 + \frac{2M}{r_+}}.\]
In \cite{n5} it is noted that if $a$ is small, and $\mu$ is small relative to $a$ and $m$, then the techniques developed by Dafermos and Rodnianski in \cite{n15} can be used to show that no unstable modes exist. Finally, unstable modes for the wave equation were ruled out in the ground-breaking mode stability work of Whiting \cite{n10}.

Though they will not directly concern us here, it is worth mentioning that there is a large literature devoted to studying \emph{quasi-normal modes}. These modes have Im$(\omega) < 0$ and satisfy different boundary conditions than the ones considered in this paper. One expects these to contain a great deal of precise information about the decay of scalar fields. See \cite{n48} for a review of the role of quasi-normal modes in the physics literature. For a sample of the mathematical study of quasi-normal modes and corresponding applications (to black hole spacetimes), we recommend \cite{n49}, \cite{n50}, \cite{n51}, \cite{n52}, \cite{n53}, \cite{n54}, \cite{n55}, \cite{n56}, \cite{n64} and the references therein.
\subsection{Black-Hole Bombs and The Physics Literature}
Soon after the discovery of superradiant wave scattering \cite{n26}, the authors of \cite{n6} speculated about placing a mirror around a black hole which would reflect superradiant frequencies. They argued that this would create a positive feedback loop and result in a ``black-hole bomb.'' Naturally, one is led to wonder if this superradiant instability can arise in a more physically natural fashion. A key breakthrough came in 1976 when Damour, Deruelle, and Ruffini observed that a good candidate is the Klein-Gordon equation with non-zero mass \cite{n22}. A few years later, Zouros and Eardley \cite{n7} and Detweiler \cite{n8} developed more involved heuristics, all leading to the same conclusion. In particular, in \cite{n7} a connection was drawn between unstable modes for the Klein-Gordon equation and the existence of bound Keplerian orbits outside the ergoregion. Furthermore, they gave some approximations for the instability rates. Various extensions/refinements, numerical and otherwise, of these results continue to appear in the physics literature, e.g. \cite{n9}, \cite{n19}, and the references therein. Many of the studies of unstable modes in the physics literature rely on the WKB approximation (\cite{n8} is an exception). Even if these WKB arguments were made rigorous, they would only become accurate as $l \to \infty$. Since our techniques are variational, no large parameter is necessary, and we produce a much more complete picture. We also remark that it is expected that small Kerr-AdS black holes should exhibit superradiant instabilities \cite{n20}, \cite{n21}.

\section{Mode Solutions}\label{modeStuff}
Before discussing the proofs of the theorems, we will provide a brief review of mode solutions and the corresponding boundary conditions.
\subsection{The Radial and Angular ODE's}
As mentioned above, we work in Boyer-Lindquist coordinates and consider solutions of the form
\[\psi(t,r,\theta,\phi) := e^{-i\omega t}e^{im\phi}S_{\kappa ml}(\theta)R(r)\]
where $\omega \in \mathbb{C},\ m \in \mathbb{Z},\ l \in \mathbb{Z}_{\geq |m|},\ \text{ and }\kappa := a^2\left(\omega^2-\mu^2\right)$. The purpose of the $\kappa$ and $l$ index will soon become clear. Our potential solution $\psi$ will actually satisfy the Klein-Gordon equation if there exists $\lambda_{\kappa ml}\in \mathbb{C}$ such that $S_{\kappa ml}$ and $R$ satisfy the following equations:
\begin{equation}\label{sml}
\frac{1}{\sin\theta}\frac{d}{d\theta}\left(\sin\theta\frac{dS_{\kappa ml}}{d\theta}\right) - \left(\frac{m^2}{\sin^2\theta} - a^2\left(\omega^2-\mu^2\right)\cos^2\theta\right)S_{\kappa ml} + \lambda_{\kappa ml}S_{\kappa ml} = 0,
\end{equation}
\begin{equation}\label{eqn}
\Delta\frac{d}{dr}\left(\Delta\frac{dR}{dr}\right) - V_{\mu}R = 0,
\end{equation}
\[V_{\mu} := -(r^2+a^2)^2\omega^2 + 4Mamr\omega - a^2m^2 + \Delta\left(\lambda_{\kappa ml} + a^2\omega^2 + \mu^2r^2\right).\]
We will refer to (\ref{sml}) as the angular ODE and (\ref{eqn}) as the radial ODE, and we allow $S_{\kappa ml}$ and $R$ to be complex.

We would like to impose boundary conditions so that $\psi$ extends smoothly to the whole spacetime and has finite energy. First of all, we must have $m \in \mathbb{Z}$. Next, when $\omega$ is real, one can show that imposing the condition that $e^{im\phi}S_{\kappa ml}(\theta)$ extends to $\mathbb{S}^2$ leads to a regular Sturm-Liouville problem. Hence, for fixed $\kappa \in \mathbb{R}$ and $m \in \mathbb{Z}$, we find a countable family of solutions $S_{\kappa ml}$ with a corresponding countable collection of eigenvalues $\lambda_{\kappa ml}$ which accumulate at $\infty$. We index the $S_{\kappa ml}$ in such a way that in the $a = 0$ case, the $S_{\kappa ml}$ are simply spherical harmonics $S_{ml}$ with eigenvalues $\lambda_{\kappa ml} = \lambda_{ml} = l(l+1)$ ($l \geq |m|$). In fact, by comparison with spherical harmonics it is not difficult to see that $\kappa \in \mathbb{R}$ implies
\begin{equation}\label{lamIneq}
\lambda_{\kappa ml} + a^2\omega^2 \geq |m|\left(|m|+1\right).
\end{equation}

In order to construct unstable modes we will need to consider $\kappa \in \mathbb{C}\setminus\mathbb{R}$. Thus, in appendix \ref{angularEig} we will show that for fixed $m$ and $l$, the eigenvalue $\lambda_{\kappa_0 ml}$ can be embedded in a holomorphic curve of eigenvalues $\lambda_{\kappa ml}$ for $\kappa$ sufficiently close to $\kappa_0$. The corresponding $S_{\kappa ml}$ will also depend holomorphically on $\kappa$.

\subsection{Boundary Conditions}\label{boundaryConditions}
We are left with the question of boundary conditions for $R$. First we need to require that $R$ extends to the horizon. Since Boyer-Lindquist coordinates break down on the horizon, we will need to change to Kerr-star coordinates (section \ref{coord}). We now ask if
\[\psi(t^*,r,\phi^*,\theta) = e^{-i\omega\left(t^*-\overline{t}(r)\right)}e^{im\left(\phi^*-\overline{\phi}(r)\right)}S_{\kappa ml}(\theta)R(r)\]
extends smoothly to $r = r_+$. This will happen if we can write
\begin{equation}\label{smoothToHorizon}
R(r) = e^{-i\left(\omega \overline{t}(r) - m\overline{\phi}(r)\right)}f(r)
\end{equation}
where $f$ extends smoothly to $r_+$. Let's define
\[\xi := \frac{i(am-2Mr_+\omega)}{r_+-r_-}.\]
It is easy to see that (\ref{smoothToHorizon}) is equivalent to
\begin{equation}\label{boundaryConditionHorizon1}
R(r) = \left(r-r_+\right)^{\xi}\rho(r)
\end{equation}
for some function $\rho$ smooth at $r_+$. The local theory in appendix \ref{localAnalysisHorizon} will show that this gives a one dimensional space of solutions to the radial ODE. Recall that we defined an $r^*$ coordinate up to a constant by
\[\frac{dr^*}{dr} := \frac{r^2+a^2}{\Delta}.\]
Then, assuming $\omega$ is real and equation (\ref{boundaryConditionHorizon1}) holds, we will have
\[\frac{dR}{dr^*} = \frac{\xi(r_+-r_-)}{2Mr_+}R + O(r-r_+) \Rightarrow \]
\begin{equation}\label{boundaryConditionHorizon2}
\frac{dR}{dr^*} + i\left(\omega - \frac{am}{2Mr_+}\right)R = O(r-r_+).
\end{equation}
Lastly, we will also need $\psi$ to have finite energy along asymptotically flat hypersurfaces. Hence, for large $r$ we must require
\begin{equation}\label{finiteEnergy}
\int_{r_++1}^{\infty}\left(\left|R\right|^2 + \left|\frac{dR}{dr}\right|^2\right)r^2dr < \infty.
\end{equation}
Let's put this all together in a definition.
\begin{defi}We say that a function $f(r): (r_+,\infty) \to \mathbb{C}$ ``satisfies the boundary conditions of a mode solution'' if
\begin{enumerate}
    \item $f(r) = (r-r_+)^{\xi}\rho(r)\text{ for a function }\rho\text{ smooth at }r_+.$
    \item $\int_{r_++1}^{\infty}\left(\left|R\right|^2 + \left|\frac{dR}{dr}\right|^2\right)r^2dr < \infty.$
\end{enumerate}
\end{defi}
Remark: The discussion in appendix \ref{loc} will imply that when $\omega \in \mathbb{R}$ and $\omega^2 > \mu^2$, the condition (\ref{finiteEnergy}) will \emph{never} be satisfied by a solution $R$ of the radial ODE; in particular, this definition is clearly not a reasonable one for studying modes ``on the real axis'' for the wave equation. We will always work in the regime $\omega^2 < \mu^2$ where there \emph{do} exist solutions of the radial ODE satisfying (\ref{finiteEnergy}).
\section{Proof of Theorem \ref{restrictions}: Restrictions on Mode Solutions Corresponding to Real $\omega$}\label{proofTheoRestrictions}
We will start with the proof of Theorem \ref{restrictions} since it is simpler than and motivates the hypotheses of Theorem \ref{blackHoleBomb}. We have placed some more technical aspects of the argument in the appendix.
\subsection{Part \ref{part1}}
Let $R$ be a solution to the radial ODE with parameters $(\omega,m,l,\mu)$ such that $\omega \in \mathbb{R}\setminus\{0\}$, $\mu^2 > \omega^2$, and $R$ satisfies the boundary conditions of a mode. We wish to show that
\[am-2Mr_+\omega = 0.\]

As reviewed in appendix \ref{loc}, an asymptotic analysis of the radial ODE shows that all solutions are either exponentially growing or exponentially decaying at infinity. Since our mode solution must have finite energy (\ref{finiteEnergy}), we conclude that $R$ is exponentially decaying at infinity.

Next, let's define the energy current,
\[Q_T := \text{Im}\left(\Delta\frac{dR}{dr}\overline{R}\right).\]
The radial ODE implies that
\[\frac{dQ_T}{dr} = 0.\]
Since $R$ must decay exponentially at infinity, we have $Q_T(\infty) = 0$. Hence, using the horizon boundary condition (\ref{boundaryConditionHorizon2}), we get
\[0 = Q_T(r_+) = (2Mr_+)\text{Im}\left(\frac{dR}{dr^*}(r_+)\overline{R(r_+)}\right) = (am-2Mr_+\omega)\left|R(r_+)\right|^2.\]
Thus, either $am-2Mr_+\omega = 0$ or $R(r_+) = 0$. However, $R(r_+) = 0$ implies that $R$ is identically $0$ (appendix \ref{loc}). We conclude that $am-2Mr_+\omega = 0$.
\subsection{Part \ref{part2}}\label{sect2}
Again we let $R$ be a solution to the radial ODE with parameters $(\omega,m,l,\mu)$ such that $\omega \in \mathbb{R}\setminus\{0\}$, $\mu^2 > \omega^2$, and $R$ satisfies the boundary conditions of a mode. From the previous section we know that we must have
\[am-2Mr_+\omega = 0.\]
We now wish to show that
\[am \neq 0.\]
Using $am-2Mr_+\omega = 0$, we may write
\begin{equation}\label{form}
V_{\mu} = -(r^2+a^2)^2\omega^2 + 4M^2\omega^2r_+(2r-r_+) + \Delta\left(\lambda_{\kappa ml} + a^2\omega^2 + r^2\mu^2\right).
\end{equation}
We now argue by contradiction. If $2Mr_+\omega = am = 0$, then
\[V_{\mu} = \Delta\left(\lambda_{0 ml} + r^2\mu^2\right)  = \Delta\left(l\left(l+1\right) + r^2\mu^2\right)\geq 0.\]
Now consider the function
\[f(r) := \text{Re}\left(\Delta\frac{dR}{dr}\overline{R}\right).\]
Since our mode solution must be exponentially decreasing at infinity, we see that $f(\infty) = 0$. The boundary conditions at the horizon imply that $f(r_+) = 0$. Hence,
\[ 0 = \int_{r_+}^{\infty}\frac{df}{dr}dr = \int_{r_+}^{\infty}\left(\Delta\left|\frac{dR}{dr}\right|^2 + \frac{V_{\mu}}{\Delta}\left|R\right|^2\right)dr.\]
This contradicts the non-triviality of $R$.
\subsection{Part \ref{part3}}
We still let $R$ be a solution to the radial ODE with parameters $(\omega,m,l,\mu)$ such that $\omega \in \mathbb{R}\setminus\{0\}$, $\mu^2 > \omega^2$, and $R$ satisfies the boundary conditions of a mode. From the previous two sections we know that
\[am-2Mr_+\omega = 0,\]
\[am \neq 0.\]

We wish to show that there exists a function $C(\omega,m,l)$ such that
\[\omega^2 < \mu^2 < \omega^2 + C(\omega,m,l).\]

Starting from (\ref{form}), using $\omega^2 = \frac{a^2m^2}{4M^2r_+^2}$, and (\ref{lamIneq}), one finds
\[\frac{dV_{\mu}}{dr}\left(r_+\right) = -4r_+\left(r_+^2+a^2\right)\omega^2 + 8M^2\omega^2r_+ + \left(r_+-r_-\right)\left(\lambda_{\kappa ml} + a^2\omega^2+r_+^2\mu^2\right)\]
\[ = 8Mr_+\omega^2\left(M-r_+\right) + \left(r_+-r_-\right)\left(\lambda_{\kappa ml} + a^2\omega^2 + r_+^2\mu^2\right)\]
\[ = \left(r_+-r_-\right)\left(-\frac{a^2m^2}{Mr_+} + \lambda_{\kappa ml} + a^2\omega^2 + r_+^2\mu^2\right)\]
\[ \geq \left(r_+-r_-\right)\left(\left|m\right|\left(\left|m\right|+1\right) - \frac{a^2m^2}{Mr_+} + r_+^2\mu^2\right) > 0.\]
In the third equality we used that $2\left(M-r_+\right) = -\left(r_+-r_-\right)$, and in the last line we used that $a < M < r_+$. Away from $r_+$, increasing $\mu$ strictly increases $V_{\mu}$, and as long as $\mu^2 > \omega^2$ the potential converge to $\infty$ as $r\to\infty$; hence, we may conclude that there exists $C(\omega,m,l)$ such that
\[\mu^2 > \omega^2 + C(\omega,m,l) \Rightarrow V_{\mu} \geq 0.\]
Now the proof concludes exactly as in section \ref{sect2}.

In order to establish that
\[\lim_{l\to\infty}C(\omega,m,l) = 0,\]
we just need Proposition \ref{eigBound3} which states
\[\frac{\partial\lambda_{\kappa ml}}{\partial \mu} > 0,\]
and the Sturm-Liouville theory fact that
\[\lim_{l\to\infty}\lambda_{\kappa ml} = \infty.\]
\section{Proof of Theorem \ref{blackHoleBomb}: Construction of Mode Solutions}
Now we will prove Theorem \ref{blackHoleBomb}. As in the previous section, we have placed some more technical aspects of the argument in the appendix.
\subsection{Outline of Proof}
Before beginning the proof we will give a brief outline. As mentioned in the introduction, we start by constructing real mode solutions. The key technical insight is a variational interpretation of real mode solutions. The variational problem will posses a degeneracy, but this will turn out to be a minor technical problem. Next, we will perturb our real mode solution into the upper complex half-plane by slightly varying $\omega$ and $\mu$. This argument relies on observing that mode solutions are in a one to one correspondence with zeros of a certain holomorphic function of $\omega$ and $\mu$. Given this, an appropriate application of the implicit function will conclude the argument. Lastly, we analyze how a mode in the upper half-plane can cross the real axis. The upshot will be that as long as we are in a bound state regime $(\mu^2 > \omega^2)$, a mode must become superradiant (proposition \ref{becomeSuper}) and lose mass (proposition \ref{loseMass}) as it enters the upper half-plane. Putting everything together will conclude the proof of Theorem \ref{blackHoleBomb}.
\subsection{Existence of Real Mode Solutions}\label{constructRealModes}
We begin with construction of modes corresponding to real $\omega$. In light of Theorem \ref{restrictions} we shall fix a choice of $\omega$, $m$, and $l$ such that $am-2Mr_+\omega = 0$ and $\omega \neq 0$. In the rest of this section all constants may depend on $\omega$, $m$, and $l$.
\subsubsection{A Variational Interpretation of Real Mode Solutions}
First, we shall need to review the local theory for the radial ODE. As recalled in appendix \ref{loc}, when $am-2Mr_+\omega = 0$, a local basis around $r_+$ of solutions to the radial ODE is given by
\[\left\{\varphi_1,\log(r-r_+)\varphi_2 + \varphi_3\right\}\]
where the $\varphi_i$ are all analytic near $r_+$, $\varphi_1(r_+) = 1$, $\varphi_2(r_+) = 1$, and $\varphi_3(r_+) = 0$. Our to be constructed solution $R$ should be a non-zero multiple of $\varphi_1$. As we already observed during the proof of Theorem \ref{restrictions}, the finite energy requirement (\ref{finiteEnergy}) and the local theory from appendix \ref{loc} implies that near infinity $R$ must be exponentially decreasing. It will be useful to further observe that the local analysis shows that if a solution of the radial ODE is not exponentially decreasing, then it is exponentially increasing.

Next, we explore the graph of $\frac{V_{\mu}}{\Delta}$. Using the formula (\ref{form}) and the assumption $\omega^2 = \frac{a^2m^2}{4M^2r_+^2}$, one may derive
\begin{equation}\label{form2}
\frac{V_{\mu}}{\Delta} = -\frac{a^2m^2}{4M^2r_+^2}\left(\Delta+4Mr + \frac{4M^2(r-r_+)}{r-r_-}\right) + \lambda_{\kappa ml} + a^2\omega^2 + r^2\mu^2.
\end{equation}
In particular, combining this with (\ref{lamIneq}) and the inequality $a < M < r_+$ gives
\[\frac{V_{\mu}}{\Delta}(r_+) =-\frac{a^2m^2}{Mr_+} + \lambda_{\kappa ml} + a^2\omega^2 + r_+^2\mu^2 \geq \]
\[-\frac{a^2m^2}{Mr_+} + \left|m\right|\left(\left|m\right|+1\right) + r_+^2\mu^2 \geq m^2\left(1- \frac{a^2}{Mr_+}\right) + r_+^2\mu^2 > 0.\]
Furthermore, it is easy to see that
\[\frac{V_{\mu}}{\Delta} = r^2\left(\mu^2-\omega^2\right) + O(r)\text{ as }r\to\infty,\]

Thus, there exists $r_+ < r_A(\mu^2) < r_B(\mu^2) < \infty$ such that $V_{\mu}$ can only be non-negative on $(r_A,r_B)$.\footnote{Note that this structure is absent in a study of the wave equation $(\mu = 0)$.} Furthermore, we can take $r_A$ increasing in $\mu^2$ and $r_B$ decreasing in $\mu^2$. Below, in lemma \ref{notEmpty} we will see that for $\mu^2$ sufficiently close to $\omega^2$, $V_{\mu}$ does in fact get very negative in $(r_A,r_B)$. This suggests that we could look for bound states of the radial ODE by minimizing the functional
\[\mathcal{L}_{\mu}(f) := \int_{r_+}^{\infty}\left(\Delta\left|\frac{df}{dr}\right|^2 + \frac{V_{\mu}}{\Delta}\left|f\right|^2\right)dr\]
over functions of unit $L^2$ norm. Note that any solution $f$ of the radial ODE with $\mathcal{L}_{\mu}(f) < \infty$ will automatically satisfy the correct boundary conditions (at $r = r_+$ \emph{and} $r = \infty$). This is the crucial way that the $am-2Mr_+\omega = 0$ assumption enters the construction. The degeneration of the radial ODE at $r_+$ poses some difficulties for a direct variational analysis of $\mathcal{L}_{\mu}$. Nevertheless, we will be able to overcome this by working with regularized versions of $\mathcal{L}_{\mu}$. In section \ref{variational} we will prove the following two propositions.
\begin{prop}\label{boundStates}For every $\mu$ sufficiently close to but larger than $\omega$, there exists a non-zero $f_{\mu}$ satisfying the boundary conditions of a mode solution and a constant $\nu_{\mu}\leq 0$ such that
\[\Delta\frac{d}{dr}\left(\Delta\frac{df_{\mu}}{dr}\right) - V_{\mu}f_{\mu} + \nu_{\mu}\Delta f_{\mu} = 0.\]
Furthermore, $\nu_{\mu}$ can be taken to be increasing in $\mu^2$.
\end{prop}
\begin{prop}\label{modeSolution}There exists $\mu_0$ and corresponding $f_{\mu_0}$ such that $\nu_{\mu_0} = 0$.
\end{prop}
The $f_{\mu_0}$ is the solution we seek.

We will close the section with a preparatory lemma. Recall that we have fixed $\omega$, $m$, and $l$ which are assumed to satisfy $am-2Mr_+\omega = 0$ and $\omega \in \mathbb{R}\setminus\{0\}$. Define
\[\mathscr{A} := \left\{\mu > 0: \mu^2 > \omega^2 \text{ and } \exists f \in C_0^{\infty}\text{ with }\mathcal{L}_{\mu}(f) < 0\right\}.\]
\begin{lemm}\label{notEmpty}Let $\mu$ be sufficiently close to but larger than $\omega$. Then we will have
\[\mu \in \mathscr{A}.\]
\end{lemm}
\begin{proof}
For every fixed $f$, $\mathcal{L}_{\mu}(f)$ is continuous in $\mu$. Thus, it is sufficient to find a smooth $f$ with compact support such that
\[\mathcal{L}_{\omega}(f) < 0.\]
First, we note that near infinity
\[\frac{V_{\omega}}{\Delta} = -2M\omega^2r + O(1).\]
Hence, for $f$ supported in $(A,\infty)$ with $A$ large, we write
\[\mathcal{L}_{\omega}(f) = \int_A^{\infty}\left(\left(r^2+O(r)\right)\left|\frac{df}{dr}\right|^2 - \left(2M\omega^2r + O(1)\right)\left|f\right|^2\right)dr.\]
Since $\omega \neq 0$, if we set $f$ to be equal to $r^{-3/4}$ on a sufficiently large compact set $K$ and $0$ outside a slight enlargement of $K$, it is clear that we can make $\mathcal{L}_{\omega}(f)$ as negative as we please.
\end{proof}
We remark that this lemma is the only place where we shall use the $\omega \neq 0$ hypothesis.
\subsubsection{Analysis of the Variational Problem}\label{variational}
It will be useful to consider the following regularization of $\mathcal{L}_{\mu}$:
\[\mathcal{L}_{\mu}^{(\epsilon)}(f) := \int_{r_++\epsilon}^{\infty}\left(\Delta\left|\frac{df}{dr}\right|^2 + \frac{V_{\mu}}{\Delta}\left|f\right|^2\right)dr.\]
\begin{lemm}\label{regularizedBoundState}If $\mu^2 > \omega^2$, then there exists $f_{\mu}^{(\epsilon)} \in H_0^1(r_++\epsilon,\infty)$ with unit $L^2(r_++\epsilon,\infty)$ norm such that $\mathcal{L}_{\mu}^{(\epsilon)}$ achieves its infimum over $H_0^1(r_++\epsilon,\infty)$ functions of unit $L^2(r_++\epsilon,\infty)$ norm on $f_{\mu}^{(\epsilon)}$.
\end{lemm}
\begin{proof}If omitted, all integration ranges are over $(r_++\epsilon,\infty)$. Recall that in section \ref{variational} we showed that $\frac{V_{\mu}}{\Delta}$ is increasing in $\mu^2$, is non-negative near $r_+$, and goes to infinity as $r \to \infty$. More specifically, we established
\[\frac{V_{\mu}}{\Delta}(r_+) \gtrsim \mu^2,\]
\[\frac{V_{\mu}}{\Delta} = r^2\left(\mu^2-\omega^2\right) + O(r) \text{ as }r\to\infty.\]
Hence, we can find $r_+ < B_0 < B_1$, $C_0 > 0$, and $C_1 > 0$ only depending\footnote{Remember that we have fixed $\omega$, $m$, and $l$ and that all constants in this section may depend on these.} on an lower bound for $\mu^2$ such that
\begin{equation}\label{trivialBound}
\int\left(\Delta\left|\frac{df}{dr}\right|^2 + C_0r^21_{[B_0,B_1]^c}\left(\mu^2-\omega^2\right)\left|f\right|^2\right)dr \leq C_1\int_{B_0}^{B_1}\left|f\right|^2dr + \mathcal{L}_{\mu}^{(\epsilon)}(f).
\end{equation}
From this it is clear that
\[\nu_{\mu}^{(\epsilon)} := \inf\left\{\mathcal{L}_{\mu}^{(\epsilon)}(f) : f\in C_c^{\infty}\text{ and }\vv f\vv_{L^2}=1\right\} > -\infty.\]

Let $\left\{f_{n,\mu}^{(\epsilon)}\right\}_{n=1}^{\infty}$ be a sequence of smooth functions, compactly supported in $(r_++\epsilon,\infty)$, with $\left\vert\left\vert f_{n,\mu}^{(\epsilon)}\right\vert\right\vert_{L^2} = 1$, such that
\[\mathcal{L}_{\mu}^{(\epsilon)}\left(f_{n,\mu}^{(\epsilon)}\right) \to \nu_{\mu}^{(\epsilon)}.\]
The bound (\ref{trivialBound}) implies that $\left\vert\left\vert f_{n,\mu}^{(\epsilon)}\right\vert\right\vert_{H^1}$ is uniformly bounded. We now apply Rellich compactness to produce a $f_{\mu}^{(\epsilon)} \in H_0^1$ such that a re-labeled subsequence of $\left\{f_{n,\mu}^{(\epsilon)}\right\}$ converges to $f_{\mu}^{(\epsilon)}$ weakly in $H^1$ and strongly in $L^2$ on compact subsets of $(r_+,\infty)$.

We claim that no mass is lost in the limit, i.e.~$\vv f_{\mu}^{(\epsilon)}\vv_{L^2} = 1$. Suppose not. Then, for any compact set $K$, there will exist infinitely many of the $f_{n,\mu}^{(\epsilon)}$'s such that
\[\left\vert\left\vert f_{n,\mu}^{(\epsilon)}\right\vert\right\vert_{L^2\left([r_+,\infty)\setminus K\right)} \geq \alpha > 0.\]
It is easy to see from (\ref{trivialBound}) that this will give a contradiction if $K$ is sufficiently large.

Using the boundedness of weak limits and the strong $L^2$ convergence, we then get
\[\nu_{\mu}^{(\epsilon)} \leq \mathcal{L}_{\mu}^{(\epsilon)}\left(f_{\mu}^{(\epsilon)}\right) \leq \liminf_{n\to\infty}\mathcal{L}_{\mu}^{(\epsilon)}\left(f_{n,\mu}^{(\epsilon)}\right) = \nu_{\mu}^{(\epsilon)}.\]
This implies that $\mathcal{L}_{\mu}^{(\epsilon)}$ achieves its minimum on $f_{\mu}^{(\epsilon)}$.
\end{proof}
Now we are ready to prove proposition \ref{boundStates}.
\begin{proof}
First we observe that $\left\{\nu_{\mu}^{(\epsilon)}\right\}_{\epsilon > 0}$ is bounded and decreasing in $\epsilon$. Set $\nu_{\mu} = \lim_{\epsilon\to 0}\nu_{\mu}^{(\epsilon)}$. Lemma \ref{notEmpty} implies that $\mu \in \mathscr{A}$ which in turn implies that $\nu_{\mu} < 0$. For any interval $K = \left(r_+ + \frac{1}{n},n\right)$ with $n$ large, (\ref{trivialBound}) implies that
\[\sup_{\epsilon > 0}\left\vert\left\vert f_{\mu}^{(\epsilon)}\right\vert\right\vert_{H^1(K)} < \infty,\]
\[\inf_{\epsilon > 0}\left\vert\left\vert f_{\mu}^{(\epsilon)}\right\vert\right\vert_{L^2(K)} > 0.\]
After an application of Rellich compactness and passing to a subsequence, we may find a non-zero $f_{\mu} \in H^1$ that is a weak $H^1$ and strong $L^2_{\text{loc}}$ limit of $f_{\mu}^{(\epsilon)}$.

Using the Euler-Lagrange equations associated to $\mathcal{L}_{\mu}^{(\epsilon)}$, we find
\[\Delta\frac{d}{dr}\left(\Delta\frac{df_{\mu}}{dr}\right) - V_{\mu}f_{\mu} + \nu_{\mu}\Delta f_{\mu} = 0.\]
On any compact subset $K$ of $(r_+,\infty)$, boundedness of weak limits and the $L^2_{\text{loc}}$ convergence of the $f_{\mu}^{(\epsilon)}$ imply that
\[\int_K\left(\Delta\left|\frac{df_{\mu}}{dr}\right|^2 + \frac{V_{\mu}}{\Delta}\left|f_{\mu}\right|^2\right) \leq \nu_{\mu}.\]
Hence,
\begin{equation}\label{aBound}
\int_{r_+}^{\infty}\Delta\left|\frac{df_{\mu}}{dr}\right|^2dr < \infty.
\end{equation}
Near $r_+$ the local theory from appendix \ref{loc} implies that
\[f_{\mu} = A\varphi_1 + B\left(\log(r-r_+)\varphi_2 + \varphi_3\right)\]
for some constants $A$ and $B$ and non-zero analytic functions $\varphi_i$. However, if $B \neq 0$, then
\[\int_{r_+}^{\infty}\Delta\left|\frac{df_{\mu}}{dr}\right|^2dr = \infty.\]
Hence, $B = 0$. Near infinity the local theory from appendix \ref{loc} implies that that $f_{\mu}$ is asymptotic to a linear combination of an exponentially growing solution and an exponentially decaying solution. The bound (\ref{aBound}) clearly implies that $f_{\mu}$ is in fact exponentially decaying. Thus, $f_{\mu}$ satisfies the boundary conditions of a mode solution.
\end{proof}
Finally, we can prove proposition \ref{modeSolution}.
\begin{proof}First we will show that $\nu_{\mu}$ is continuous for $\mu \in \mathscr{A}$. Let us normalize each $f_{\mu}^{(\epsilon)}$ so that $\left\vert\left\vert f_{\mu}^{(\epsilon)}\right\vert\right\vert_{L^2} = 1$. We have
\[\nu_{\mu_1}^{(\epsilon)} = \mathcal{L}_{\mu_1}^{(\epsilon)}\left(f_{\mu_1}^{(\epsilon)}\right) = \mathcal{L}_{\mu_2}^{(\epsilon)}\left(f_{\mu_1}^{(\epsilon)}\right) + \int_{r_++\epsilon}^{\infty}\frac{V_{\mu_1}-V_{\mu_2}}{\Delta}\left|f_{\mu_1}^{(\epsilon)}\right|^2dr \geq \] \[\nu_{\mu_2}^{(\epsilon)} - \left|\mu_1^2 - \mu_2^2\right|\int_{r_++\epsilon}^{\infty}r^2\left|f_{\mu_1}^{(\epsilon)}\right|^2dr.\]
Reversing the roles of $\mu_1$ and $\mu_2$ gives
\[\left|\nu_{\mu_1}^{(\epsilon)}-\nu_{\mu_2}^{(\epsilon)}\right| \leq \left|\mu_1^2-\mu_2^2\right|\int_{r_++\epsilon}^{\infty}r^2\left(\left|f_{\mu_1}^{(\epsilon)}\right|^2 + \left|f_{\mu_2}^{(\epsilon)}\right|^2\right)dr \leq\]
\[\left|\mu_1^2-\mu_2^2\right|\left(C + \int_{r_++\epsilon}^{\infty}r^21_{[B_0,B_1]^c}\left(\left|f_{\mu_1}^{(\epsilon)}\right|^2 + \left|f_{\mu_2}^{(\epsilon)}\right|^2\right)dr\right) \leq\]
\[C'\left|\mu_1^2-\mu_2^2\right|.\]
In these inequalities we have used (\ref{trivialBound}), $\left\vert\left\vert f_{\mu}^{(\epsilon)}\right\vert\right\vert_{L^2} = 1$,  and the fact that the $\nu_{\mu_i}^{(\epsilon)}$ are negative. Since the constant $C'$ is independent of $\epsilon$, we may take $\epsilon$ to $0$.

By lemma \ref{notEmpty} $\mathscr{A} \neq \emptyset$. Hence, we may set
\[\mu_0 := \sup\mathscr{A}.\]
It is clear that for any $\mu \in \mathscr{A}$, we cannot have $V_{\mu}$ strictly positive on $(r_+,\infty)$. Thus $\mu_0 < \infty$. Since $\nu_{\mu}$ is increasing in $\mu$, we may extend $\nu_{\mu}$ continuously so that $\nu_{\mu_0}$ exists. We will of course have $\nu_{\mu_0} \leq 0$. Suppose that $\nu_{\mu_0} < 0$. Then, one may easily show that $\mu_0 \in \mathscr{A}$, and hence we can run the existence argument above to construct a corresponding $f_{\nu_{\mu_0}}$. Next, by continuity we could slightly increase $\mu_0$ to $\mu'_0 \in \mathscr{A}$, run the existence argument again, and conclude that $\nu_{\mu'_0} < 0$. This of course contradicts the definition of $\mu_0$. We conclude that $\nu_{\mu_0} = 0$.

It remains to show that there exists a corresponding $f_{\mu_0}$. From the local theory in appendix \ref{loc}, for every $\mu$ and $\nu$, we have a unique solution $\tilde{R}(r,\mu,\nu)$ to
\[\Delta\frac{d}{dr}\left(\Delta\frac{d\tilde{R}}{dr}\right) - V_{\mu}\tilde{R} + \nu\Delta\tilde{R} = 0\]
which satisfies $\tilde{R}(r_+,\mu,\nu) = 1$. At infinity there will be a local basis of solutions spanned by $\tilde{\rho}_1(r,\mu,\nu)$ and $\tilde{\rho}_2(r,\mu,\nu)$ where $\tilde{\rho_1}$ is exponentially increasing, $\tilde{\rho}_2$ is exponentially decreasing, and both depend analytically on $r$, $\mu$, and $\nu$. Lastly, we have analytic reflection and transmission coefficients $\tilde{A}(\mu,\nu)$ and $\tilde{B}(\mu,\nu)$ defined by
\[\tilde{R}(r,\mu,\nu) = \tilde{A}(\mu,\nu)\tilde{\rho}_1(r,\mu,\nu) + \tilde{B}(\mu,\nu)\tilde{\rho}_2(r,\mu,\nu).\]
As $\mu \uparrow \mu_0$ we have $\tilde{A}(\mu,\nu_{\mu}) = 0$. It follows that $\tilde{A}(\mu_0,0) = 0$. We may then set
\[f_{\mu_0}(r) := \tilde{R}(r,\mu_0,0).\]
\end{proof}
\subsection{Construction of the Exponentially Growing Modes}\label{growModes}
In this section our goal is to perturb the real modes into the complex upper half-plane with an appropriate application of the implicit function theorem. Using the previous section we may start with a solution $R$ to the radial ODE satisfying the boundary conditions of a mode and with frequency parameters $(\omega_R(0),m,l,\mu(0))$ such that $\omega_R(0) \in \mathbb{R}$ and $\mu^2(0) > \omega_R^2(0)$. For any $\omega = \omega_R + i\omega_I$ and $\mu$ sufficiently close to $\omega_R(0)$ and $\mu(0)$ respectively, the local theory from appendix \ref{loc} will give us two linearly independent solutions to the radial ODE $\rho_1(r,\omega,\mu)$ and $\rho_2(r,\omega,\mu)$ such that $\rho_1$ is exponentially increasing at infinity, $\rho_2$ is exponentially decreasing at infinity, and both depend holomorphically on $\omega$ and analytically on $\mu$. Furthermore, the local theory around $r_+$ tells us that, up to normalizing properly, for each $\omega = \omega_R + i\omega_I$ and $\mu$ we have a unique local solution $R(r,\omega,\mu)$ around $r_+$ satisfying the correct boundary condition. We have
\begin{equation}\label{connect}
R(r,\omega,\mu) =  A(\omega,\mu)\rho_1(r,\omega,\mu) + B(\omega,\mu)\rho_2(r,\omega,\mu).
\end{equation}
As shown in appendix \ref{loc}, $A$ and $B$ are holomorphic in $\omega$ and $\mu$. Finding a mode solution is equivalent to finding a zero of $A$. We have picked our parameters so that $A\left(\omega_R(0),\mu(0)\right)) = 0$. Let's write $A = A_R + iA_I$. Next, we note that an application of the implicit function theorem will produce our unstable modes if we can establish
\[
\det\begin{pmatrix}
    \frac{\partial A_R}{\partial \omega_R} & \frac{\partial A_R}{\partial \mu} \\
    \frac{\partial A_I}{\partial \omega_R} & \frac{\partial A_I}{\partial \mu}
\end{pmatrix}\left(\omega_R(0),\mu(0)\right) \neq 0.
\]
In order to do this, we shall return to the energy current
\[Q_T = \text{Im}\left(\Delta\frac{dR}{dr}\overline{R}\right).\]
Recall that in section \ref{proofTheoRestrictions} we saw
\[\frac{dQ_T}{dr} = 0,\]
\[Q_T(r_+) = am-2Mr_+\omega_R.\]
We have used the normalization $\left|R(r_+)\right|^2 = 1$ in the second statement. Next, let's write $Q_T$ in terms of $\rho_1$ and $\rho_2$.
\[Q_T = \left|A\right|^2\Delta\text{Im}\left(\frac{d\rho_1}{dr}\overline{\rho_1}\right) + \Delta\text{Im}\left(A\frac{d\rho_1}{dr}\overline{B\rho_2}\right) +\] \[\Delta\text{Im}\left(B\frac{d\rho_2}{dr}\overline{A\rho_1}\right) + \left|B\right|^2\Delta\text{Im}\left(\frac{d\rho_2}{dr}\overline{\rho_2}\right).\]
Before examining this at infinity, let us note the precise asymptotics of the $\rho_i$ as recalled in appendix \ref{loc}.
\[\rho_1 \sim e^{\sqrt{\mu^2-\omega_R^2}r}r^{-1 + \frac{M(2\omega_R^2-\mu^2)}{\sqrt{\mu^2-\omega_R^2}}},\]
\[\rho_2 \sim e^{-\sqrt{\mu^2-\omega_R^2}r}r^{-1 - \frac{M(2\omega_R^2-\mu^2)}{\sqrt{\mu^2-\omega_R^2}}}.\]
Furthermore, it's easy to see from the construction of the $\rho_i$ that they are both real valued. Now let's compute $Q_T(\infty)$. Since the $\rho_i$ are real, the first and last terms clearly vanish. The exponential powers cancel in the middle terms, and we find
\[Q_T(\infty) = \sqrt{\mu^2-\omega_R^2}\text{Im}\left(A\overline{B}\right) -\sqrt{\mu^2-\omega_R^2}\text{Im}\left(B\overline{A}\right) = 2\sqrt{\mu^2-\omega_R^2}\text{Im}\left(A\overline{B}\right).\]
We conclude that
\begin{equation}\label{eqn2}
am-2Mr_+\omega_R = 2\sqrt{\mu^2-\omega_R^2}\text{Im}\left(A\overline{B}\right).
\end{equation}
Since $A(\omega_R(0),\mu(0)) = 0$, taking derivatives of (\ref{eqn2}) implies that
\[-2Mr_+ = 2\sqrt{\mu^2(0)-\omega_R^2(0)}\text{Im}\left(\frac{\partial A}{\partial \omega_R}\left(\omega_R(0),\mu(0)\right)\overline{B}\left(\omega_R(0),\mu(0)\right)\right),\]
\[0 = 2\sqrt{\mu^2(0)-\omega_R^2(0)}\text{Im}\left(\frac{\partial A}{\partial \mu}\left(\omega_R(0),\mu(0)\right)\overline{B}\left(\omega_R(0),\mu(0)\right)\right).\]
Since $B\left(\omega_R(0),\mu(0)\right) \neq 0$, these two equations imply that the vectors $\left(\frac{\partial A_R}{\partial \omega_R},\frac{\partial A_I}{\partial \omega_R}\right)$ and $\left(\frac{\partial A_R}{\partial \mu},\frac{\partial A_I}{\partial\mu}\right)$ are linearly independent at $\left(\omega_R(0),\mu(0)\right)$ if and only if $\frac{\partial A}{\partial \mu}\left(\omega_R(0),\mu(0)\right) \neq 0$, i.e.
\[
\det\begin{pmatrix}
    \frac{\partial A_R}{\partial \omega_R} & \frac{\partial A_R}{\partial \mu} \\
    \frac{\partial A_I}{\partial \omega_R} & \frac{\partial A_I}{\partial \mu}
\end{pmatrix}\left(\omega_R(0),\mu(0)\right) \neq 0 \Leftrightarrow \frac{\partial A}{\partial \mu}\left(\omega_R(0),\mu(0)\right) \neq 0.
\]
It remains to establish
\begin{lemm}\label{nonDegen}
\[\frac{\partial A}{\partial \mu} \neq 0.\]
\end{lemm}
\begin{proof}
For the sake of contradiction, suppose that $\frac{\partial A}{\partial \mu}\left(\omega_R(0),\mu(0)\right) = 0$. Differentiating (\ref{connect}) gives
\[\frac{\partial R}{\partial \mu}(r,\omega_R(0),\mu(0)) =  \frac{\partial B}{\partial \mu}(\omega_R(0),\mu(0))\rho_2(r,\omega_R(0),\mu(0)) + \] \[B(\omega_R(0),\mu(0))\frac{\partial\rho_2}{\partial\mu}(r,\omega_R(0),\mu(0)).\]
This implies that $\frac{\partial R}{\partial \mu}$ is exponentially decreasing at infinity. The analysis from appendices \ref{localAnalysis} and \ref{loc} implies that $\frac{\partial R}{\partial \mu}$ is smooth at $r_+$.\footnote{Recall that, as discussed in appendix \ref{loc}, $R$ is smooth at $r_+$ when $am-2Mr_+ = 0$.} Differentiating the radial ODE with respect to $\mu$, multiplying by $\overline{R}$, and integrating gives
\begin{equation}\label{something}
\int_{r_+}^{\infty}\left(\frac{\partial}{\partial r}\left(\Delta\frac{\partial^2R}{\partial r\partial \mu}\right) - \frac{V_{\mu}}{\Delta}\frac{\partial R}{\partial\mu}\right)\overline{R}dr = \int_{r_+}^{\infty}\left(2\mu r^2 + \frac{\partial\lambda_{\kappa ml}}{\partial\mu}\right)\left|R\right|^2dr.
\end{equation}
Integrating by parts twice on the left hand side will produce no boundary terms since both $R$ and $\frac{\partial R}{\partial \mu}$ are exponentially decreasing at infinity and $\Delta\left(r_+\right) = 0$. Thus we have
\[\int_{r_+}^{\infty}\left(\frac{\partial}{\partial r}\left(\Delta\frac{\partial^2R}{\partial r\partial \mu}\right) - \frac{V_{\mu}}{\Delta}\frac{\partial R}{\partial\mu}\right)\overline{R}dr = \]
\[\int_{r_+}^{\infty}\frac{\partial R}{\partial\mu}\left(\frac{\partial}{\partial r}\left(\Delta\frac{\partial \overline{R}}{\partial r}\right) - \frac{V_{\mu}}{\Delta}\overline{R}\right)dr = 0.\]
We have used that $\overline{R}$ is a solution of the radial ODE in the last equality. Plugging this into (\ref{something}) then gives
\[\int_{r_+}^{\infty}\left(2\mu r^2 + \frac{\partial\lambda_{\kappa ml}}{\partial\mu}\right)\left|R\right|^2dr = 0.\]
Since proposition \ref{eigBound3} from the appendix says that
\[\frac{\partial\lambda_{\kappa ml}}{\partial\mu} > 0,\]
we conclude that $R$ vanishes, a contradiction.
\end{proof}
\subsection{Modes Crossing the Real Axis}\label{modeCross}
In this section we shall investigate how a mode can ``cross'' the real axis. Let's introduce a little more notation. From the analysis of the previous section we have a family of solutions $R(r,\epsilon)$ to the radial ODE satisfying the boundary conditions of a mode with parameters $(\omega(\epsilon),m,l,\mu(\epsilon))$ where $\omega(\epsilon) = \omega_R(\epsilon) + i\epsilon$. Implicitly we have also been using the existence of a family $\lambda_{\kappa ml}$ of eigenvalues to the angular ODE (see proposition \ref{eigCurves}). These functions are all defined for $\left|\epsilon\right| \ll 1$. In what follows we will often omit the $\epsilon$'s and we shall assume $0 < \omega_R(\epsilon) < \mu(\epsilon)$. Using the symmetry of the equations under $(\omega,m) \mapsto (-\omega,-m)$ one may check that this assumption implies no loss of generality. The function $R$ will satisfy
\[\frac{\partial}{\partial r}\left(\Delta\frac{\partial R}{\partial r}\right) - \frac{V_{\mu}}{\Delta}R = 0,\]
\[R \sim (r-r_+)^{\frac{i(am-2Mr_+\omega)}{r_+-r_-}}\text{ at }r_+,\]
\[R \sim e^{-r\sqrt{\mu^2-\omega^2}}r^{-1 - \frac{M(2\omega^2-\mu^2)}{\sqrt{\mu^2-\omega^2}}}\text{ at }r = \infty,\]
\[V_{\mu} := -(r^2+a^2)^2\omega^2 + 4Mamr\omega - a^2m^2 + \Delta\left(\lambda + a^2\omega^2 + \mu^2r^2\right).\]
We also have
\[\frac{1}{\sin\theta}\frac{\partial}{\partial\theta}\left(\sin\theta\frac{\partial S}{\partial\theta}\right) - \left(\frac{m^2}{\sin^2\theta} - a^2\left(\omega^2-\mu^2\right)\cos^2\theta\right)S + \lambda S = 0\]
where $S(\cdot,\epsilon): \theta \in (0,\pi) \to \mathbb{C}$ is given boundary conditions which make it regular at $\theta = 0,\pi$ (see appendix \ref{angularEig}). Note that we have suppressed the $\kappa$, $m$, and $l$ indices from $S_{\kappa ml}$ and $\lambda_{\kappa ml}$.

From Theorem \ref{restrictions} we know that
\begin{equation}\label{thresholdSuper}
\omega_R(0) = \frac{am}{2Mr_+}.
\end{equation}
We wish to investigate the signs of $\frac{\partial \omega_R}{\partial\epsilon}(0)$ and $\frac{\partial \mu}{\partial\epsilon}(0)$. The condition \ref{thresholdSuper} corresponds to our mode solution being exactly on the threshold of superradiance. This makes sense because when $\epsilon = 0$ the solution neither grows nor decays with time. For $\epsilon > 0$ the mode solution will grow with time. Hence, we expect the mode to become superradiant (\ref{sup}).

This leads us to
\begin{prop}\label{becomeSuper}If $\epsilon > 0$ we must have
\[\omega^2_R(\epsilon) + \epsilon^2 < \left(\frac{am}{2Mr_+}\right)^2.\]
In particular
\[\frac{\partial\omega_R}{\partial\epsilon}\left(0\right) \leq 0.\]
\end{prop}
\begin{proof}We now introduce a variant of the microlocal energy current $Q_T$:
\[\tilde{Q}_T := \text{Im}\left(\Delta\frac{\partial R}{\partial r}\overline{\omega R}\right).\]
For $\epsilon > 0$ we have
\[\tilde{Q}_T(\infty) = \tilde{Q}_T(r_+) = 0\]
\[\frac{\partial \tilde{Q}_T}{\partial r} = -\epsilon\Delta\left|\frac{\partial R}{\partial r}\right|^2 + \text{Im}\left(\frac{V_{\mu}\overline{\omega}}{\Delta}\right)\left|R\right|^2 \Rightarrow \]
\[\int_{r_+}^{\infty}\left(\epsilon\Delta\left|\frac{\partial R}{\partial r}\right|^2 - \text{Im}\left(\frac{V_{\mu}\overline{\omega}}{\Delta}\right)\left|R\right|^2\right)dr = 0.\]
We have
\[\text{Im}\left(-V_{\mu}\overline{\omega}\right) = \epsilon\left((r^2+a^2)^2\left|\omega\right|^2 - a^2m^2 + \Delta r^2\mu^2\right) - \Delta\text{Im}\left(\left(\lambda+a^2\omega^2\right)\overline{\omega}\right).\]
Proposition \ref{eigBound1} from the appendix gives
\[-\text{Im}\left(\left(\lambda+a^2\omega^2\right)\overline{\omega}\right) > 0.\]
Furthermore, $\text{Im}\left(-V_{\mu}\overline{\omega}\right)$ is increasing in $r$. Thus, $R \neq 0$ implies that
\[\text{Im}\left(-V_{\mu}\overline{\omega}\right)(r_+) < 0 \Leftrightarrow \]
\[\omega^2_R(\epsilon) + \epsilon^2 < \left(\frac{am}{2Mr_+}\right)^2.\]
\end{proof}
\begin{prop}\label{eigProp}
\[\frac{\partial\mu}{\partial\epsilon}(0) \geq 0 \Rightarrow \text{Re}\left(\frac{\partial\lambda}{\partial\epsilon}\right)(0) \geq 0.\]
\end{prop}
\begin{proof}
Let's set
\[S_{\epsilon} := \frac{\partial S}{\partial \epsilon}.\]
At $\epsilon = 0$ we have
\[\frac{1}{\sin\theta}\frac{\partial}{\partial\theta}\left(\sin\theta\frac{\partial S_{\epsilon}}{\partial\theta}\right) - \left(\frac{m^2}{\sin^2\theta} - a^2\left(\omega^2-\mu^2\right)\cos^2\theta\right)S_{\epsilon} + \lambda S_{\epsilon} = \]
\[-\left(2a^2\cos^2\theta\left(\omega_R\left(\frac{\partial\omega_R}{\partial\epsilon} + i\right) - \mu\frac{\partial\mu}{\partial\epsilon}\right) + \frac{\partial\lambda}{\partial\epsilon}\right)S.\]
Using appendix \ref{localAnalysis} one may check that $S_{\epsilon}$ is regular at $\theta = 0,\pi$. Multiplying by $\overline{S}$ and integrating by parts implies
\[\int_0^{\pi}\left(2a^2\cos^2\theta\left(\omega_R\left(\frac{\partial\omega_R}{\partial\epsilon}+i\right)- \mu\frac{\partial\mu}{\partial\epsilon}\right)+\frac{\partial\lambda}{\partial\epsilon}\right)\left|S\right|^2\sin\theta d\theta = 0.\]
Using proposition \ref{becomeSuper} we conclude that
\[\frac{\partial\mu}{\partial\epsilon}(0) \geq 0 \Rightarrow \text{Re}\left(\frac{\partial\lambda}{\partial\epsilon}\right)(0) \geq 0.\]
\end{proof}
Finally, we examine $\frac{\partial \mu}{\partial \epsilon}(0)$.
\begin{prop}\label{loseMass}
\[\frac{\partial\mu}{\partial\epsilon}\left(0\right) < 0.\]
\end{prop}
\begin{proof}Let's set
\[R_{\epsilon} := \frac{\partial R}{\partial\epsilon}.\]
We have
\begin{equation}\label{eqn3}
\frac{\partial}{\partial r}\left(\Delta\frac{\partial R_{\epsilon}}{\partial r}\right) - \frac{V_{\mu}}{\Delta}R_{\epsilon} = \frac{\partial}{\partial\epsilon}\left(\frac{V_{\mu}}{\Delta}\right)R.
\end{equation}
Now we want to multiply by $\overline{R}$ and integrate by parts. However, we have to be careful with regards to $R_{\epsilon}'s$ boundary conditions. At infinity $R_{\epsilon}$ may easily be seen to be exponentially decreasing, but at $r_+$ the proper condition is more subtle. By construction
\[(r-r_+)^{\frac{-i(am-2Mr_+\omega(\epsilon))}{r_+-r_-}}R(r,\epsilon) =: G(r,\epsilon)\]
is analytic in $r$ and $\epsilon$ near $(r_+,0)$. At $\epsilon = 0$ we have
\[(r-r_+)^{\frac{-i(am-2Mr_+\omega)}{r_+-r_-}}R_{\epsilon} - \]
\[\frac{2Mr_+}{r_+-r_-}\left(1 - i\frac{\partial\omega_R}{\partial\epsilon}\right)(r-r_+)^{\frac{-i(am-2Mr_+\omega)}{r_+-r_-}}\log(r-r_+)R = \frac{\partial G}{\partial\epsilon}\Rightarrow \]
\[R_{\epsilon}(r,0) = \frac{2Mr_+}{r_+-r_-}\left(1-i\frac{\partial\omega_R}{\partial\epsilon}\right)\log(r-r_+)R(r,0) + \frac{\partial G}{\partial\epsilon}(r,0).\]
Now we multiply (\ref{eqn3}) by $\overline{R}$, take the real part, and integrate by parts. We end up with
\[-2Mr_+\left|R(r_+)\right|^2 = \int_{r_+}^{\infty}\text{Re}\left(\frac{\partial}{\partial\epsilon}\left(\frac{V_{\mu}}{\Delta}\right)\right)\left|R\right|^2dr = \]
\[\int_{r_+}^{\infty}\Delta^{-1}\left(2\omega_R\frac{\partial\omega_R}{\partial\epsilon}\left(-\Delta^2+(a^2-4Mr)\Delta - 4M^2r(r-r_+)\right)\right)\left|R\right|^2dr + \]
\[\int_{r_+}^{\infty}\left(\text{Re}\left(\frac{\partial\lambda}{\partial\epsilon}\right) + 2r^2\mu\frac{\partial\mu}{\partial\epsilon}\right)\left|R\right|^2dr.\]
Now proposition \ref{eigProp} finishes the proof.
\end{proof}
\section{Following the Unstable Modes in the Upper Half Plane}\label{intoHalfPlane}
Following our construction of unstable modes near the real axis, it is natural to ask if one can continue to decrease $\mu$ and produce more unstable modes. We will not explore this in detail in this paper, but we will briefly describe the expected behavior. One believes that one can vary $\mu$ and produce a $1$-parameter (at least continuous) family of modes with frequency parameters $(\omega(\mu),\mu,m,l)$.\footnote{A potential approach is to more directly exploit the underlying analyticity, see \cite{n45} and \cite{n48} for ideas along these lines.} As long as these modes are in the upper half plane, proposition \ref{becomeSuper} shows that they will remain superradiant. Hence, if and when they cross the real axis, they will satisfy $|\omega| \leq \frac{|am|}{2Mr_+}$. Now, note that proposition \ref{loseMass} implies that in the bound state regime $(\mu^2 > \omega^2)$, an unstable mode can cross the real axis only by increasing the mass. Hence, as long as we decrease $\mu$ and maintain $\mu > \frac{|am|}{2Mr_+}$, the curve of modes cannot cross the real axis, and, by continuity, we conclude that these modes would have to remain unstable.
\section{Acknowledgements}
I would like to thank Igor Rodnianski, Mihalis Dafermos, Gustav Holzegel, and Jonathan Luk for useful conversations and for advice on the exposition.
\appendix
\section{Linear ODE's with Regular Singularities}\label{localAnalysis}
Let's recall some facts about linear ODEs in the complex plane.
\begin{lemm}\label{regularSingLocal}Consider the complex ODE
\begin{equation}\label{ODE}
\frac{d^2H}{dz^2} + f(z,\lambda)\frac{dH}{dz} + g(z,\lambda)H = 0.
\end{equation}
We will assume that there exists $\{f_j(\lambda)\}$, $\{g_j(\lambda)\}$, $r$, and open $U \subset\mathbb{C}$ such that for $z \in B_r(z_0)$ and every compact $K \subset U$ there exists $\left\{F_j^{(K)}\right\}$ and $\left\{G_j^{(K)}\right\}$ such that
\[\left|f_j(\lambda)\right| \leq F_j^{(K)}\text{ and }\left|g_j(\lambda)\right| \leq G_j^{(K)}\text{ when }\lambda \in K,\]
\[\sum_{j=0}^{\infty}G_j^{(K)}(z-z_0)^j\text{ and }\sum_{j=0}^{\infty}F_j^{(K)}(z-z_0)^j\text{ converge absolutely},\]
\[\{f_j(\lambda)\}\text{ and }\{g_j(\lambda)\}\text{ are holomorphic in }\lambda \in U,\]
\[(z-z_0)f(z,\lambda) = \sum_{j=0}^{\infty}f_j(\lambda)(z-z_0)^j\text{ and }(z-z_0)^2g(z) = \sum_{j=0}^{\infty}g_j(\lambda)(z-z_0)^j.\]
If these hypotheses hold we say that $z_0$ is a regular singularity.
Set
\[Q(\alpha,\lambda) := \alpha(\alpha - 1) + f_0(\lambda)\alpha + g_0(\lambda).\]
The indicial equation is
\[Q(\alpha,\lambda) = 0.\]
We suppose that a holomorphic $\alpha(\lambda)$ has been chosen such that
\[Q(\alpha(\lambda),\lambda) = 0\text{ and }\min_{j \in \mathbb{Z}^+}\left|Q(\alpha(\lambda) + j,\lambda)\right| = A(\lambda) > 0.\]
Then there exists a unique solution to (\ref{ODE}) of the form
\[h(z,\lambda) = (z-z_0)^{\alpha(\lambda)}\rho(z,\lambda)\]
such that $\rho(z_0,\lambda) = 1$. Furthermore, $\rho$ is holomorphic for $z \in B_r(z_0)$ and $\lambda \in U$.
\end{lemm}
\begin{proof}One can extract a proof of this from the discussion of regular singularities in \cite{n1}. For the sake of completeness we will give the needed slight extension here. Without loss of generality we may set $z_0 = 0$. We begin by looking for a formal solution of the form
\[h(z,\lambda) = z^{\alpha(\lambda)}\sum_{j=0}^{\infty}\rho_j(\lambda)z^j\]
where we set $\rho_0(\lambda) = 1$. Formally plugging this into (\ref{ODE}) we find (see \cite{n1})
\[Q\left((\alpha(\lambda),\lambda)\right) = 0,\]
\[Q\left(\alpha(\lambda) + j,\lambda\right)\rho_j(\lambda) = -\sum_{k=0}^{j-1}\left(\left(\alpha(\lambda)+k\right)f_{j-k}(\lambda) + g_{j-k}(\lambda)\right)\rho_k(\lambda)\text{ for }j \geq 1.\]
Since $Q(\alpha(\lambda),\lambda) = 0$ by hypothesis, the first equation is satisfied. Furthermore, by assumption $Q(\alpha(\lambda) + j,\lambda) \neq 0$ for any $j$. Hence, the second equation determines $\rho_j(\lambda)$ recursively. This establishes the uniqueness of $\rho$. It remains to check that the series converges appropriately. We will do this by majorizing the series. Let us pick an arbitrary compact set $K \subset U$ and $r_0 < r$. After applying Cauchy's estimate to the holomorphic functions $\sum_j F_j^{(K)}z^j$ and $\sum_j G_j^{(K)}z^j$, we may find a constant $C_K$ so that
\[\left|f_j(\lambda)\right| \leq C_Kr_0^{-j}\text{ and }\left|g_j(\lambda)\right| \leq C_Kr_0^{-j}\text{ for }\lambda \in K.\]
Let $\beta(\lambda)$ be the other root of $Q(\cdot,\lambda)$, and set $n(\lambda) := \left|\alpha(\lambda) - \beta(\lambda)\right|$. Since $Q(\alpha(\lambda) + k,\lambda) = k(k+\alpha(\lambda)-\beta(\lambda))$, our hypotheses imply that $\alpha(\lambda)-\beta(\lambda) \not\in \mathbb{Z}_{\leq0}$. Next, define $b_j(\lambda)$ by
\[b_j(\lambda) = \left|\rho_j(\lambda)\right|\text{ for }j \leq n,\]
\[j(j-\left|\alpha(\lambda)-\beta(\lambda)\right|)b_j(\lambda) = C_K\sum_{k=0}^{j-1}\left(|\alpha(\lambda)| + k+1\right)b_k(\lambda)r_0^{k-j}\text{ for }j > n.\]
It is easy to check by induction that $\left|\rho_j(\lambda)\right| \leq b_j(\lambda)$ for all $j$. For sufficiently large $j$, one finds that
\[r_0j(j-|\alpha(\lambda)-\beta(\lambda)|)b_j(\lambda) - (j-1)(j-1-|\alpha(\lambda)-\beta(\lambda)|)b_{j-1}(\lambda) = \] \[C_K(|\alpha(\lambda)|+j)b_{j-1}(\lambda).\]
Now the ratio test implies that the series $\sum_{j\geq 0}^{\infty}b_j(\lambda)z^j$ converges in the ball of radius $r_0$. Hence, by the comparison test, $\sum_{j= 0}^{\infty}\rho_j(\lambda)z^j$ converges in the same ball. Since $r_0$ was arbitrary, we find that for every $\lambda \in K$, $\sum_{j=0}^{\infty}\rho_j(\lambda)z^j$ converges and is holomorphic in $z \in B_r(0)$. Next we may freeze $z \in B_r(0)$ and consider $\rho(z,\lambda) = \sum_{j=0}^{\infty}\rho_j(\lambda)z^j$ as a function of $\lambda$. For every compact $K \subset U$, our proof has shown that $\rho(z,\cdot)$ is a uniform limit of holomorphic functions. Hence, $\rho(z,\lambda)$ is holomorphic for $\lambda \in U$.
\end{proof}
\section{The Angular ODE}\label{angularEig}
In this section we will establish the needed facts about the eigenvalues of the angular ODE. We assume throughout this section that $m \neq 0$.

Recall that $\kappa := a^2(\omega^2-\mu^2)$. Then the angular ODE is
\[\frac{1}{\sin\theta}\frac{d}{d\theta}\left(\sin\theta\frac{dS}{d\theta}\right) - \left(\frac{m^2}{\sin^2\theta} - \kappa\cos^2\theta\right)S + \lambda S = 0.\]
We have suppressed the $\kappa$, $m$, and $l$ indices.

\begin{prop}\label{eigCurves}
Suppose that for some fixed $\kappa_0 \in \mathbb{R}$ we have an eigenvalue $\lambda_0$. Then, for $\kappa$ sufficiently close to $\kappa_0$, we can uniquely find a holomorphic curve $\lambda(\kappa)$ of eigenvalues for the angular ODE with parameter $\kappa$ such that $\lambda_0 = \lambda(\kappa_0)$.
\end{prop}
\begin{proof}Let's change variables to $x := \cos\theta$. Then the angular ODE becomes
\[\frac{d}{dx}\left((1-x^2)\frac{dS}{dx}\right) - \left(\frac{m^2}{1-x^2} - \kappa x^2\right)S + \lambda S = 0\text{ with }x \in (-1,1).\]
An asymptotic analysis (appendix \ref{localAnalysis}) at $x = \pm 1$ shows that any solution must be asymptotic to a linear combination of $\left(1\mp x\right)^{|m|/2}$ and $\left(1\mp x\right)^{-|m|/2}$ as $x \to \pm 1$. If $S$ is an eigenfunction we clearly must have
\[S \sim (1\pm x)^{|m|/2}\text{ as }x\to \pm 1.\]
For any $\kappa$ and $\lambda$ we can uniquely define a solution $S(\theta,\kappa,\lambda)$ by requiring that
\begin{equation}\label{expansion}
S(\theta,\kappa,\lambda)(1+x)^{-|m|/2}\text{ is holomorphic at }x = -1,
\end{equation}
\[\left(S(\cdot,\kappa,\lambda)(1+\cdot)^{-|m|/2}\right)(x = -1) = 1.\]
We then have holomorphic functions $F(\kappa,\lambda)$ and $G(\kappa,\lambda)$ such that
\[S(\theta,\kappa,\lambda) \sim F(\kappa,\lambda)(1-x)^{-|m|/2} + G(\kappa,\lambda)(1-x)^{|m|/2}\text{ as }x\to 1.\]
Since $\lambda_0$ is an eigenvalue, we have $F(\kappa_0,\lambda_0) = 0$. We will be able to uniquely define our curve $\lambda(\kappa)$ for $\kappa$ near $\kappa_0$ via an application of the implicit function theorem if we can verify that
\[\frac{\partial F}{\partial \lambda}\left(\kappa_0,\lambda_0\right) \neq 0.\]
For the sake of contradiction, assume that
\[\frac{\partial F}{\partial \lambda}\left(\kappa_0,\lambda_0\right) = 0.\]
Set
\[S_{\lambda} := \frac{\partial S}{\partial\lambda}.\]
By differentiating (\ref{expansion}) and using that $F$ and $\frac{\partial F}{\partial\lambda}$ vanish at $(\kappa_0,\lambda_0)$, one may easily check that $S_{\lambda}$ still satisfies the boundary conditions of an eigenfunction. It will also satisfy
\[\frac{d}{dx}\left((1-x^2)\frac{dS_{\lambda}}{dx}\right) - \left(\frac{m^2}{1-x^2} - \kappa_0 x^2\right)S_{\lambda} + \lambda_0 S_{\lambda} = -S.\]
Multiplying both sides of this equation by $\overline{S}$, integrating over $(0,\pi)$, and then integrating by parts will imply that
\[\int_0^{\pi}|S|^2\sin\theta d\theta = 0.\]
This is clearly a contradiction.
\end{proof}
\begin{prop}\label{eigBound1}
\[\omega_I > 0 \Rightarrow \text{Im}\left(\left(\lambda + a^2\omega^2\right)\overline{\omega}\right) < 0.\]
\end{prop}
\begin{proof}Multiplying the angular ODE by $\overline{\omega S}$, integrating by parts, and taking imaginary parts gives
\[\omega_I\int_0^{\pi}\left(\left|\frac{dS}{d\theta}\right|^2 + \left(\frac{m^2}{\sin^2\theta} + a^2\left|\omega\right|^2\sin^2\theta + a^2\mu^2\cos^2\theta\right)\left|S\right|^2\right)\sin\theta d\theta  = \] \[-\int_0^{\pi}\text{Im}\left(\left(\lambda+a^2\omega^2\right)\overline{\omega}\right)\left|S\right|^2\sin\theta d\theta.\]
\end{proof}
\begin{prop}\label{eigBound3}When $\omega$ is real, we have
\[\frac{\partial\lambda}{\partial\mu} > 0.\]
\end{prop}
\begin{proof}Let
\[S_{\mu} := \frac{\partial S}{\partial\mu}.\]
We have
\[\frac{1}{\sin\theta}\frac{d}{d\theta}\left(\sin\theta\frac{dS_{\mu}}{d\theta}\right) - \left(\frac{m^2}{\sin^2\theta} - a^2(\omega^2-\mu^2)\cos^2\theta\right)S_{\mu} + \lambda S_{\mu} = \]
\[\left(2a^2\mu\cos^2\theta - \frac{\partial\lambda}{\partial\mu}\right)S.\]
Multiplying the equation by $\overline{S}$, integrating by parts, and taking the real part gives
\[\int_{0}^{\pi}\left(2a^2\mu\cos^2\theta - \frac{\partial\lambda}{\partial\mu}\right)\left|S\right|^2\sin\theta d\theta = 0.\]
\end{proof}
\section{Local Theory for the Radial ODE}\label{loc}
\subsection{The Horizon}\label{localAnalysisHorizon}
Let's apply the theory from appendix \ref{localAnalysis} to the radial ODE. Recall that we earlier set
\[\xi = \frac{i(am-2Mr_+\omega)}{r_+-r_-}.\]
First we consider the case where $am-2Mr_+\omega \neq 0$. In this case the indicial equation has two distinct roots which do not differ by an integer. Hence a local basis of solutions to the radial ODE around $r_+$ will be given by
\[\left\{(r-r_+)^{\xi}\rho_1(r),(r-r_+)^{-\xi}\rho_2(r)\right\}\]
where each $\rho_i(r)$ is holomorphic near $r_+$ and is normalized to have $\rho_i(r_+) = 1$. Our mode analysis from section \ref{boundaryConditions} showed that a mode solution must be of the form $A(r-r_+)^{\xi}\rho_1(r)$ for some $A \in \mathbb{C}$. Hence, for every $\omega$ and $\mu$ so that $\lambda$ is defined, we have a unique solution to the radial ODE of the form
\begin{equation}\label{localSolutions}
(r-r_+)^{\xi}\rho(r,\omega,\mu)
\end{equation}
where $\rho(r,\omega,\mu)$ is analytic in $r$, holomorphic in $\omega$, analytic in $\mu$, and $\rho(r_+,\omega,\mu) = 1$. Let us remark that if a mode solution with real $\omega$ and $am-2Mr_+\omega \neq 0$ vanishes at $r_+$, it must vanish identically.

Now let's consider what happens if $am-2Mr_+\omega = 0$. In this case the indicial equation has a double root at $\alpha = 0$ and lemma \ref{regularSingLocal} only produces one solution near $r_+$. One must then consider solutions which have a logarithmic singularity at $r_+$. The standard theory (see \cite{n1}) then implies that a local basis of solutions is given by
\[\left\{\varphi_1(r),\log(r-r_+)\varphi_2(r) + \varphi_3(r)\right\}\]
where the $\varphi_i$ are all holomorphic near $r_+$, $\varphi_1(r_+) = 1$, $\varphi_2(r_+) = 1$, and $\varphi_3(r_+) = 0$. It will be important to note that lemma \ref{regularSingLocal} implies that $\varphi_1$ is embedded in the family of local solutions \ref{localSolutions}.

Lastly, it will be useful for the bound state analysis to note that everything said in this section so far applies verbatim to the equation
\[\Delta\frac{d}{dr}\left(\Delta\frac{dR}{dr}\right) - V_{\mu}R - \nu\Delta R = 0\text{ for }\nu \in \mathbb{R}.\]
\subsection{Infinity}
The local existence theorems quoted in this section can be found in chapter 7 of \cite{n1}. Let us note that the radial ODE can be written as
\[\frac{dR^2}{dr^2} + \frac{\partial_r\Delta}{\Delta}\frac{dR}{dr} - \frac{V_{\mu}}{\Delta^2}R = 0\Rightarrow \]
\[\frac{dR^2}{dr^2} + \left(\frac{2}{r} + O(r^{-2})\right)\frac{dR}{dr} + \left((\omega^2-\mu^2) + \frac{2M(2\omega^2-\mu^2)}{r} + O(r^{-2})\right)R = 0.\]
Let's write $\omega = \omega_R + i\omega_I$. We will need to construct a local basis at infinity that depends holomorphically on $\omega$ and analytically on $\mu$.
\begin{lemm}For all $\omega$ and $\mu$ with $\mu^2 - \omega^2 \not\in (-\infty,0]$ there is a unique $\rho_1(r,\omega,\mu)$ which solves the radial ODE and satisfies
\[\rho_2(r,\omega,\mu) = e^{-\sqrt{\mu^2-\omega^2}r}r^{-1 - \frac{M(2\omega^2-\mu^2)}{\sqrt{\mu^2-\omega^2}}} + O\left(e^{-\sqrt{\mu^2-\omega^2}r}r^{-2 - \frac{M(2\omega^2-\mu^2)}{\sqrt{\mu^2-\omega^2}}}\right).\]
Furthermore, $\rho_2$ depends holomorphically on $\omega$ and $\mu$. The square root is defined by making a branch cut along the negative real numbers.
\end{lemm}
\begin{proof}One can more or less extract a proof of this from the discussion of irregular singularities in Chapter 7 section 2 of \cite{n1}. For the sake of completeness we will give the needed slight extension. We let $C$ denote a sufficiently large constant which can be taken holomorphic in $\mu$ and $\omega$. One may find a formal solution to the radial ODE of the form
\[L(r,\omega,\mu) := e^{-\sqrt{\mu^2-\omega^2}r}r^{-1-\frac{M(2\omega^2-\mu^2)}{\sqrt{\mu^2-\omega^2}}}\sum_{j=0}^{\infty}\frac{a_j(\omega,\mu)}{z^j}\]
where $a_0 = 1$ and the $a_j$ are holomorphic in $\omega$ and $\mu$. See Chapter 7 section 1 of \cite{n1} for the computations behind this. Let's set
\[L_n(r,\omega,\mu) := e^{-\sqrt{\mu^2-\omega^2}r}r^{-1-\frac{M(2\omega^2-\mu^2)}{\sqrt{\mu^2-\omega^2}}}\sum_{j=0}^{n-1}\frac{a_j(\omega,\mu)}{z^j}.\]
Then
\[\frac{d^2L_n}{dr^2} + \frac{\partial_r\Delta}{\Delta}\frac{dL_n}{dr} - \frac{V_{\mu}}{\Delta^2}L_n =  e^{-\sqrt{\mu^2-\omega^2}r}r^{-1-\frac{M(2\omega^2-\mu^2)}{\sqrt{\mu^2-\omega^2}}}B_n(r,\omega,\mu) \]
where $B_n(r,\omega,\mu) \leq Cr^{-n-1}$.
Let's look for a solution $\rho_2$ of the form
\[\rho_2(r,\omega,\mu) = L_n(r,\omega,\mu) + \epsilon(r,\omega,\mu).\]
We must have
\[\frac{d^2\epsilon}{dr^2} + \frac{\partial_r\Delta}{\Delta}\frac{d\epsilon}{dr} - \frac{V_{\mu}}{\Delta^2}\epsilon = - e^{-\sqrt{\mu^2-\omega^2}r}r^{-1-\frac{M(2\omega^2-\mu^2)}{\sqrt{\mu^2-\omega^2}}}B_n \Leftrightarrow \]
\[\frac{d^2\epsilon}{dr^2} + \left(\omega^2-\mu^2\right)\epsilon =\]
\[- e^{-\sqrt{\mu^2-\omega^2}r}r^{-1-\frac{M(2\omega^2-\mu^2)}{\sqrt{\mu^2-\omega^2}}}B_n - \frac{\partial_r\Delta}{\Delta}\frac{d\epsilon}{dr} + \left(\frac{V_{\mu}}{\Delta^2} + \left(\omega^2-\mu^2\right)\right)\epsilon.\]
Let's set
\[K(r,t) := \frac{e^{\sqrt{\mu^2-\omega^2}(r-t)}-e^{-\sqrt{\mu^2-\omega^2}(r-t)}}{2\sqrt{\mu^2-\omega^2}}.\]
Variation of parameters gives
\[\epsilon(r,\omega,\mu) = \]
\[\int_r^{\infty}K(r,t)\left( e^{-\sqrt{\mu^2-\omega^2}t}t^{-1-\frac{M(2\omega^2-\mu^2)}{\sqrt{\mu^2-\omega^2}}}B_n(t) - \left(\frac{V_{\mu}(t)}{\Delta^2(t)}+\omega^2-\mu^2\right)\epsilon(t) + \frac{\partial_t\Delta(t)}{\Delta(t)}\frac{d\epsilon}{dr}(t)\right)dt.\]
We may solve this by iterating in the usual fashion. Set $h_0(r,\omega,\mu) = 0$ and
\[h_{j+1}(r,\omega,\mu) = \]
\[\int_r^{\infty}K(r,t)\left( e^{-\sqrt{\mu^2-\omega^2}t}t^{-1-\frac{M(2\omega^2-\mu^2)}{\sqrt{\mu^2-\omega^2}}}B_n - \left(\frac{V_{\mu}(t)}{\Delta^2(t)}+\omega^2-\mu^2\right)h_j(t) + \frac{\partial_t\Delta(t)}{\Delta(t)}\frac{dh_j}{dr}(t)\right)dt.\]
It is easy to see that
\[\left|h_1(r,\omega,\mu)\right| + \left|\frac{dh_1}{dr}(r,\omega,\mu)\right| \leq \]
\[\frac{C e^{-\sqrt{\mu^2-\omega^2}r}r^{-1-\frac{M(2\omega^2-\mu^2)}{\sqrt{\mu^2-\omega^2}}}}{r^n}\left(n+\frac{M(2\omega^2-\mu^2)}{\sqrt{\mu^2-\omega^2}}\right)^{-1}.\]
Then, with induction one can show that
\[\left|h_{j+1}-h_j\right|(r,\omega,\mu) + \left|\frac{dh_{j+1}}{dr} - \frac{dh_j}{dr}\right|(r,\omega,\mu) \leq \]
\[\frac{C^je^{-\sqrt{\mu^2-\omega^2}r}r^{-1-\frac{M(2\omega^2-\mu^2)}{\sqrt{\mu^2-\omega^2}}}}{r^n}\left(n + \frac{M(2\omega^2-\mu^2)}{\sqrt{\mu^2-\omega^2}}\right)^{-j}.\]
For $\omega$ and $\mu$ in a sufficiently small compact set and sufficiently large $n$, the $h_j(r,\omega,\mu)$ will converge uniformly in $r$, $\omega$, and $\mu$.
\end{proof}
It is of course easy to pick a second holomorphic family of solutions $\rho_1(r,\omega,\mu)$ that is linearly independent of $\rho_2$. One can show (Chapter 7 of \cite{n1}) that we must then have
\[\rho_2(r,\omega,\mu) \sim e^{-\sqrt{\mu^2-\omega^2}r}r^{-1 - \frac{M(2\omega^2-\mu^2)}{\sqrt{\mu^2-\omega^2}}},\]
\[\rho_1(r,\omega,\mu) \sim e^{\sqrt{\mu^2-\omega^2}r}r^{-1 + \frac{M(2\omega^2-\mu^2)}{\sqrt{\mu^2-\omega^2}}}.\]
Lastly, we note that a similar discussion can be carried out for the equation
\[\Delta\frac{d}{dr}\left(\Delta\frac{dR}{dr}\right) - V_{\mu}R + \nu\Delta R = 0.\]
\subsection{Reflection and Transmission Coefficients}
Let's fix some set of frequency parameters with $\mu^2 - \omega_R^2 \not\in (-\infty,0]$. Above we constructed $\rho(r,\omega,\mu)$ holomorphic in $\omega$ and $\mu$ so that $(r-r_+)^{\xi}\rho(r,\omega,\mu)$ gives a solution to the radial ODE with the correct boundary condition at $r_+$. We can then introduce reflection and transmission coefficients $A(\omega,\mu)$ and $B(\omega,\mu)$:
\[R(r,\omega,\mu) := (r-r_+)^{\xi}\rho(r,\omega,\mu) = A(\omega,\mu)\rho_1(r,\omega,\mu) + B(\omega,\mu)\rho_2(r,\omega,\mu).\]
Let $W(\cdot,\cdot)$ denote the Wronskian. Then
\[A = \frac{W(R,\rho_2)}{W(\rho_1,\rho_2)}.\]
Thus $A$ is holomorphic in $\omega$ and analytic $\mu$. Similarly, $B$ is holomorphic in $\omega$ and analytic in $\mu$.

\end{document}